%% file: arxiv.tex
\documentclass{article}

\usepackage[preprint,nonatbib]{neurips_2023}

\usepackage[utf8]{inputenc} 
\usepackage[T1]{fontenc}    
\usepackage{url}            
\usepackage{booktabs}       
\usepackage{amsfonts}       
\usepackage{nicefrac}       
\usepackage{microtype}      
\usepackage{xcolor}         

\usepackage[vlined,ruled,linesnumbered]{algorithm2e}
\usepackage{subcaption}
\usepackage{csquotes}
\usepackage{multirow}
\usepackage{amsmath,amsfonts,amsthm,color,mathtools}
\usepackage{url,cleveref}
\usepackage{color,soul}

\newtheorem{definition}{Definition}

\newtheorem{theorem}{Theorem}

\newcommand{\ob}{\mathbb{O}}
\newcommand{\V}{\mathcal{V}}
\newcommand{\E}{\mathcal{E}}

 \usepackage{enumitem}
\def\bitcoin{%
  \leavevmode
  \vtop{\offinterlineskip 
    \setbox0=\hbox{B}%
    \setbox2=\hbox to\wd0{\hfil\hskip-.03em
    \vrule height .3ex width .15ex\hskip .08em
    \vrule height .3ex width .15ex\hfil}
    \vbox{\copy2\box0}\box2}} 
\AtBeginDocument{%
  \providecommand\BibTeX{{%
    \normalfont B\kern-0.5em{\scshape i\kern-0.25em b}\kern-0.8em\TeX}}}
    
\def\ra{0.9}

\title{Chainlet Orbits: Topological Address Embedding for the Bitcoin Blockchain}

 \author{%
  Poupak Azad\\
  Department of Computer Science\\
  University of Manitoba\\
  Winnipeg, MB \\
  \texttt{azadp@myumanitoba.ca} \\
    \And
    Baris Coskunuzer \\
    Department of Mathematical Sciences\\
  University of Texas at Dallas \\
  Richardson, TX, USA \\
  \texttt{coskunuz@utdallas.edu} \\
  \AND
  Murat Kantarcioglu \\
   Department of Computer Science \\
   University of Texas at Dallas \\
    Richardson, TX, USA \\
  \texttt{muratk@utdallas.edu} \\
    \And
  Cuneyt Gurcan Akcora\\
  Departments of Computer Science and Statistics\\
  University of Manitoba \\
  Winnipeg, MB, Canada \\
  \texttt{cuneyt.akcora@umanitoba} \\
}

\newcommand{\C}{\mathcal{C}}
\newcommand{\I}{\mathcal{I}}

\newcommand{\B}{\mathcal{B}}
\newcommand{\G}{\mathcal{G}}
\newcommand{\h}{\mathcal{H}}
\newcommand{\A}{\mathcal{A}}

\newcommand{\s}{\mathcal{S}}

\newcommand{\wh}{\widehat} 
\newcommand{\wt}{\widetilde}

\begin{document}

\maketitle
 
\begin{abstract}
The rise of cryptocurrencies like Bitcoin, which enable transactions with a degree of pseudonymity, has led to a surge in various illicit activities, including ransomware payments and transactions on darknet markets. These illegal activities often utilize Bitcoin as the preferred payment method. However, current tools for detecting illicit behavior either rely on a few heuristics and laborious data collection processes or employ computationally inefficient graph neural network (GNN) models that are challenging to interpret.

To overcome the computational and interpretability limitations of existing techniques, we introduce an effective solution called Chainlet Orbits. This approach embeds Bitcoin addresses by leveraging their topological characteristics in transactions. By employing our innovative address embedding, we investigate e-crime in Bitcoin networks by focusing on distinctive substructures that arise from illicit behavior. 

The results of our node classification experiments demonstrate superior performance compared to state-of-the-art methods, including both topological and GNN-based approaches. Moreover, our approach enables the use of interpretable and explainable machine learning models in as little as 15 minutes for most days on the Bitcoin transaction network.
\end{abstract}

\section{Introduction}
The success of Bitcoin~\cite{nakamoto2008bitcoin} has also encouraged significant usage of cryptocurrencies for illegal activities ranging from ransomware payments~\cite{akcora2021bitcoinheist} to darknet market transactions~\cite{kethineni2018use}. For instance, the CryptoLocker ransomware spreads through spam emails, and once installed, it establishes communication with a command-and-control center, encrypts the system resources, and requests Bitcoin payments to unlock the encrypted resources. 

Since the Bitcoin transaction network can be represented as a heterogeneous directed graph, we can apply graph machine learning techniques to detect illicit transactions in node or edge classification tasks~\cite{zola2019cascading,lorenz2020machine}. Among the  techniques used, graph neural networks (GNNs) (e.g.,  GCN~\cite{kipf2016semi}, DGCNN~\cite{zhang2018end}, GIN~\cite{xu2018powerful} and GraphSage~\cite{hamilton2017inductive}) allow powerful node embeddings by first summarizing the substructure information of node neighborhoods for each node as vectors, and then combining them to get a representation of the whole graph. While GNNs are quite powerful methods, for very large graphs, they fail to embed the crucial domain information in a computationally efficient manner. Furthermore, in many cases, the results are not interpretable and explainable due to complex aggregations and embeddings done by the GNNs. 

Alternatively, orbit analysis can be employed to examine the graph geometry surrounding a node, leading to the creation of interpretable behavioral patterns~\cite{lorrain1971structural}. In the realm of graph theory and network analysis, an orbit denotes a unique configuration or arrangement of nodes and edges within a graph that recurs throughout its structure. This notion captures the symmetry and repetitive patterns present in the graph. Blockchain transaction networks lend themselves well to orbit analysis due to their heterogeneous nature~\cite{akcora2022blockchain}, which organizes nodes into clearly defined substructures, namely transactions.

We define a notion of {orbits} to determine the common structural patterns between nodes in Bitcoin transaction graphs (e.g., see Figure~\ref{fig:rsmostfreq}). Using orbit-based node embedding, we study e-crime in Bitcoin networks by focusing on special substructures induced by the illicit behavior and show that we can efficiently match the existing GNN detection accuracy while allowing interpretable and explainable machine learning models. To the best of our knowledge, \textbf{this is the first work to provide search and query capabilities over UTXO graphs}. Our framework is well-suited for human-in-the-loop scenarios, particularly when there is a need for manual analysis of numerous Bitcoin addresses due to compliance requirements.

 \begin{figure}[] 
    \centering
    \includegraphics[width=0.7\linewidth]{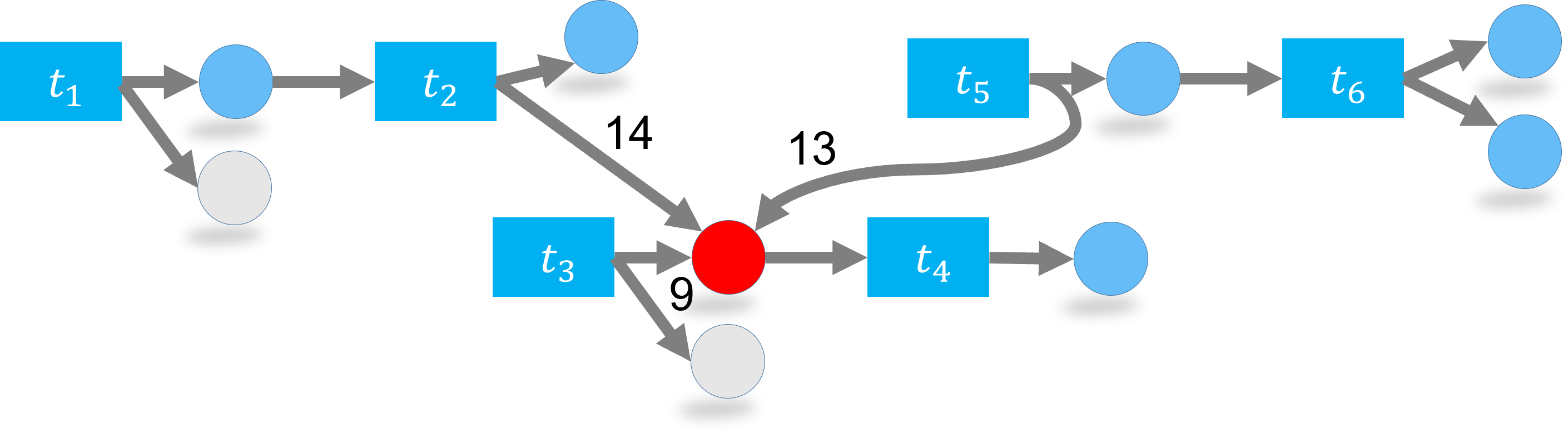}
    \caption{By analyzing transaction patterns within the directed Bitcoin network, we identify specific address (circular nodes) behaviors in transactions (rectangular nodes) and assign them structural roles which we refer to as orbits. In this example, a red (e-crime) address has been identified and assigned orbits 9, 13, and 14, which are commonly used by addresses associated with ransomware.}
    \label{fig:rsmostfreq}
\end{figure}

Our key contributions can be summarized as follows:
\begin{itemize}
\item \textit{Efficient Orbit Discovery}: We present a topological-equivalence-based node embedding called "orbit" that can effectively handle very large directed graphs, and give the mathematical foundations of the theory by applying group actions on graphs. 
\item \textit{E-Crime Address Discovery}: We distinguish orbits that capture the transaction behaviors of e-crime participants, enabling an analyst to focus on the actions of specific groups of addresses such as those controlled by ransomware operators.
\item \textit{Extensive Empirical Evaluation}: We conduct node classification experiments on the entire Bitcoin transaction network, and  demonstrate that orbit-based models outperform existing topological, and graph neural network-based methods. \footnote{Our implementation is available at \url{https://github.com/chainletRepo/chainlet}.}
\item \textit{Interpretable and Explainable  Models in Graph ML}: We use orbit patterns to profile typical behavior and enable searches and similarity queries on all blockchains similar to the Bitcoin transaction graph. 
\end{itemize}

In section~\ref{sec:prelim}, we discuss the necessary background related to our research. In section~\ref{sec:chainlet}, we define the chainlet and orbit notations that are used for the efficient embedding of the node information on Bitcoin graphs.  In section~\ref{sec:experiments}, we report extensive experimental evaluation results on real-life datasets. Finally, in section~\ref{sec:conc}, we conclude by discussing the implications of our results and future work.

\section{Background and Preliminaries}
\label{sec:prelim}
In this section, we introduce preliminaries on Blockchain, Bitcoin, and recent Blockchain research. Afterward, we introduce graph machine learning on blockchain networks and discuss node classification methods.

\subsection{E-crime on Blockchain}
The success of Bitcoin~\cite{nakamoto2008bitcoin} has also encouraged significant usage of cryptocurrencies for illegal activities. 
Bitcoin provides pseudo-anonymity; although all transactions are public by nature, user identification is not required to join the network.

The earliest results aimed at tracking the transaction network to locate bitcoins used in illegal activities, such as money laundering and blackmailing (e.g., \cite{androulaki2013evaluating}), by using heuristics. At the same time, mixing schemes (e.g., ~\cite{maxwell2013coinjoin,ruffing2014coinshuffle})
have been developed to hide the flow of coins in the network. Earlier research results have shown that some bitcoins can be traced~\cite{meiklejohn2013fistful} in the network. As a result, obfuscation efforts~\cite{narayanan2017obfuscation} by malicious users have become increasingly sophisticated. This in return increased the need for illegal activity detection techniques that go beyond simple heuristics. 
 
\noindent\textbf{Ransomware}. Ransomware is a type of malware that infects a victim's computer systems, IoT devices, and mobile devices and demands payment in exchange for releasing control over the infected resources~\cite{martin2018depth}. Ransomware can either lock access to resources or encrypt their contents. It can be delivered through various means, such as email attachments or web-based vulnerabilities, and more recently through mass exploit campaigns. For instance, the CryptoLocker ransomware used the Gameover ZeuS botnet to spread through spam emails. Once the ransomware is installed, it establishes communication with a command-and-control center. While earlier variants of ransomware used fixed IP addresses or domain names, newer versions may utilize anonymity networks, such as TOR, to connect to a hidden command-and-control server.

Montreal~\cite{paquet2018ransomware}, Princeton~\cite{huang2018tracking} and Padua~\cite{conti2018economic} studies have analyzed ransomware networks on cryptocurrencies and found that hacker behavior can help us identify undisclosed ransomware payments. The datasets of these three studies are publicly available. 

Early studies in ransomware detection have used decision rules on amounts and times of known ransomware transactions to locate undisclosed ransomware (CryptoLocker) payments~\cite{liao2016behind}. More recent studies are joint efforts between researchers and Blockchain analytics companies. For example, Huang et al.~\cite{huang2018tracking} identify shared hacker behavior and use heuristics to identify ransomware payments. The authors estimate that 20,000 victims have made ransomware payments. Nevertheless, these studies utilize proprietary features for identifying ransomware payments and families, and some studies do not disclose their feature definitions~\cite{AML}.  

\noindent\textbf{Darknet Markets}. Darknet markets are online platforms that can only be accessed through the use of privacy-enhancing tools such as TOR and I2P. These markets enable sellers to create online stores and connect with buyers. The merchandise offered on these markets ranges from illicit goods such as drugs and weapons to legal items such as fake passports. Darknet markets employ an escrow service, where buyers are provided with a one-time use address to make payments for goods. After payment, the market instructs the seller to physically ship the item to the buyer. Once the item is delivered, the market releases the coins to the seller's account, who can then cash out the coins. Our analysis in \Cref{sec:experiments} corroborates the delay results; we found that most market payments remain in an address for some days.

 Cryptocurrencies are ideal for darknet markets as payments can be made pseudo-anonymously and can be received anywhere in the world. Starting with the Silk Road in 2011, darknet markets have been processing an ever-growing amount of illicit goods. Law enforcement agencies have seized multiple darknet markets, but the ease of setting up new markets has turned this into a cat-and-mouse game. Typically, Bitcoin and Monero are used for Darknet payments. In multiple cases, the seized markets' data have been made public, where each sale is recorded with information such as the buyer, seller, good, amount, and time. The Darknet Market Archives, which we use in our experiments, is a collection of such released sales~\cite{gwerndarknet}.

\subsection{Blockchain Transaction Network}
A blockchain is an immutable public ledger that records transactions in discrete data structures called \textit{blocks}~\cite{nakamoto2008bitcoin}. The earliest blockchains are cryptocurrencies such as Bitcoin and Litecoin where a transaction is a transfer of coins.

A \textit{UTXO} (Unspent Transaction Output) network, such as Bitcoin, is a directed, weighted multigraph where an edge denotes coin transfers.

 We model the Bitcoin UTXO network with a weighted, directed heterogeneous graph with address and transaction nodes. We create the graph $\mathcal{G} = (\V, \E, \B)$ from a set of transactions $TX$ and input and output addresses in $TX$. On $\mathcal{G}$, $\V$ is a set of nodes, and $\E \subseteq \V\times \V$ is a set of edges. $\B=\left\{ \textbf{Address}, \textbf{Transaction} \right\}$ represents the set of node types. A node $u \in \V$ has a node type $\phi(u) \in B$.

For each edge $e_{uv} \in \E$ between adjacent nodes $u$ and $v$, we have $\phi(u) \neq \phi(v)$. In other words, there is no edge between the same node types (Transaction$\to$Transaction or Address$\to$Address), and an edge $e \in \E$ represents a coin transfer between an address node and a transaction node. This heterogeneous graph model subsumes the homogeneous case (i.e., $\left | B \right |=1$), where only transaction or address nodes are used, and edges link nodes of the same type. Here, we focus on the case where each address node is linked (i.e., input or output address of a transaction) via a transaction node to another address node. We use $\Gamma_a^i$ and $\Gamma_a^o$ to refer to predecessors (in-neighbors) and successors (out-neighbors) of an address $a$, respectively.
On the heterogeneous Bitcoin network, the in-neighbors $\Gamma_n^{i}$ of a transaction $tx_n$ is defined as the set of transactions (not addresses) whose one or more outputs are input to transaction $tx_n$. The out-neighbors of $tx_n$ are denoted as $\Gamma_n^{o}$. 

A transaction has inputs and outputs; the sum of output amounts of a transaction $tx_n$ is defined as $\mathcal{A}^{o}(n) =\sum\limits_{a_u \in \Gamma_n^o}{{A}_u^{o}(n)}$, where ${A}_u^{o}(n)$ is the amount of coins the output address $a_u$ receives from transaction $tx_n$. $\A^i(n)$ is defined similarly, and $\A^i(n)=\A^o(n)$, i.e., the total input amount is equal to the total output amount for any transaction $tx_n$.

On the Bitcoin network, an address may appear multiple times with different inputs and outputs. An address $u$ that appears in a transaction at time $t$ can be denoted as $a_u^t$. To mine address behavior in time, we divide the Bitcoin network into 24-hour long windows by using the UTC-6 timezone as a reference. This window approach serves two purposes. First, the induced 24-hour network allows us to capture how fast a coin moves in the network. The speed is measured by the number of blocks in the 24-hour window that contains a transaction involving the coin. Second, temporal information of transactions, such as the local time, has been found useful to cluster criminal transactions  (see Figure 7 in
\cite{huang2018tracking}). 

\subsection{Blockchain Graph Analysis}
Early efforts in Blockchain graph analysis used heuristics based on transaction behavior to de-anonymize and associate Bitcoin addresses~\cite{meiklejohn2013fistful}. Two heuristics by Meiklejohn et al.~\cite{meiklejohn2013fistful} are widely used in Bitcoin address clustering: \textbf{Co-spending} and \textbf{Transition}. Co-spending posits that "if two addresses are inputs to the same transaction, they are controlled by the same user". Transition posits that "if we observe one transaction with addresses A and B as inputs, and another with addresses B and C as inputs, then we conclude that A, B, and C all belong to the same user". We did not implement the third heuristic named \textbf{Change} used in some settings, because the results of the change heuristics are not definitive; addresses discovered using this heuristic are not necessarily controlled by the same user.

As the number of e-crime transactions increased on cryptocurrencies, algorithmic models beyond simple heuristics have been developed to link bitcoin addresses by using node features on the network~\cite{huang2018tracking}.

Mark et al.~\cite{AML} introduced the Elliptic dataset, which is a temporal graph dataset consisting of over 200,000 Bitcoin node transactions, 234,000 payment edges, and 49 transaction graphs with unique time steps. Each transaction was categorized as "licit," "illicit," or "unknown." They employed machine learning techniques including Logistic Regression (LR), Random Forest (RF), Graph Convolutional Networks (GCNs) \cite{kipf2016semi}, and EvolveGCN \cite{pareja2020evolvegcn} for transaction labeling. The evaluation showed that RF achieved a recall score of 0.67 in the illicit category, while GCNs achieved a score of 0.51. Elliptic dataset only contains anonymized transactions and omits address nodes, which we need to extract Orbits. As a result, we could not apply our Orbit approach to the dataset. Furthermore, in our experiments, we show that using a neighborhood approach yields better results as measured by AUC.

Vassallo et al. \cite{vassallo2021application} detected illicit cryptocurrency activities such as scams, terrorism financing, and Ponzi schemes, and used the Elliptic dataset to evaluate its effectiveness at both the address and transaction classification.

Feature extraction on the Bitcoin transaction network has many examples. Chaehyeon et al. \cite{lee2019toward} employed supervised machine learning algorithms to classify e-crime addresses in the Bitcoin network. The authors extracted classification features based on the characteristics of Bitcoin transactions, such as transaction fees and size.  Monamo et al. \cite{monamo2016unsupervised} used nearest-neighbor methods to detect fraud in the Bitcoin network, using graph centrality measures and currency features, such as the total amount sent. \textit{Our method avoids domain-specific feature engineering}, and creates \textit{structural topological node embeddings that can be applied to any UTXO-based blockchain} (See the appendix for applicability to other blockchains).

Graph neural networks achieve better results and do not require feature engineering. AlArab et al.~\cite{Alarab} proposed a model that combines GCNs and linear layers to classify illicit nodes in the Elliptic dataset~\cite{AML}. The results showed that the model had an overall classification accuracy of 97.40\% and a recall of 0.67 for detecting illicit transactions. In \cite{nan2018bitcoin}, the authors used an autoencoder with graph embedding to detect mixing and demixing services for Bitcoin cryptocurrency, and the proposed model effectively performed the detection of these anomalies. In \cite{pham2016anomaly}, the authors used unsupervised learning to identify suspicious nodes in the Bitcoin transaction graph. Unlike these approaches, \textit{our method can operate on larger networks and provide interpretable and explainable results}.

\section{Chainlets and  Orbits}
\label{sec:chainlet}

In this section, we first outline UTXO chainlets and give mathematical background, then define our topological address embeddings, i.e., orbits. We start with giving our intuition and motivation.

\noindent\textbf{Intuition behind Orbits:} UTXO transaction networks are composed of transaction and address nodes, as shown in Figure \ref{fig:twochainlet}. The position of address nodes within a subgraph (i.e., chainlet) determines their structural role, referred to as an "orbit," which can be leveraged in supervised learning tasks. Our primary objective is to establish these roles and utilize them in graph machine learning applications.

\noindent\textbf{Chainlet}.  
Chainlets have been initially introduced in the setting of temporal predictions to define position invariant versions of occurrence and amount information of Bitcoin temporal transaction patterns (see overviews by~\cite{akcora2021blockchain,akcora2017chainlet}).  The key idea of the chainlet approach is to offer a lossless summary of all transactions by assigning two indices $i,o$ in the $[0,n]$ range to each transaction with respect to its input and output addresses. A Bitcoin subgraph $\G'=(\V', \E',\B)$ is a \textit{subgraph} of $\G$, if $\V'\subseteq \V$ and $\E'\subseteq \E$. If $\mathcal{G'}=(\V', \E',\B)$ is a subgraph of $\mathcal{G}$ and $\E'$ contains all edges $e_{uv}\in \E$ such that $(u,v) \in V'$, then $\G'$ is called an \textit{induced subgraph} of $\G$. Two graphs $\mathcal{G'}=(\V', \E',\B)$ and $\mathcal{G^{''}}=(\V^{''}, \E^{''},\B)$ are called \textit{isomorphic} if there exists a bijection $h: \V' \to \V^{''}$ such that all node pairs $u,v\in\V'$ are adjacent  in $\G'$ if and only if $h(u)$ and $h(v)$
are adjacent in $\G^{''}$.

\begin{definition} [$k$-chainlet]
Let \textit{$k$-chainlet} $\C_k$ in $\G$ be the subgraph of $\G$ containing $k$ subsequent transaction nodes, all output nodes of these $k$ transaction nodes, and all input nodes of the first transaction  (See \Cref{fig:twochainlet}).
\end{definition}

A 1-chainlet lists a transaction node with its input and output addresses. \Cref{fig:twochainlet} shows a 2-chainlet; transactions $t_1$ and $t_2$ share the common address $a_1$. A $k$-chainlet where $k>1$ has a natural ordering of transactions because transactions are ordered in blocks which, in turn, are ordered in a blockchain. This natural order allows us to assign roles to an address depending on where the address appears in a 2-chainlet. In the next section, we will count address roles and list them as address orbits.

\begin{figure}[h]
    \centering
    \includegraphics[width=0.6\linewidth]{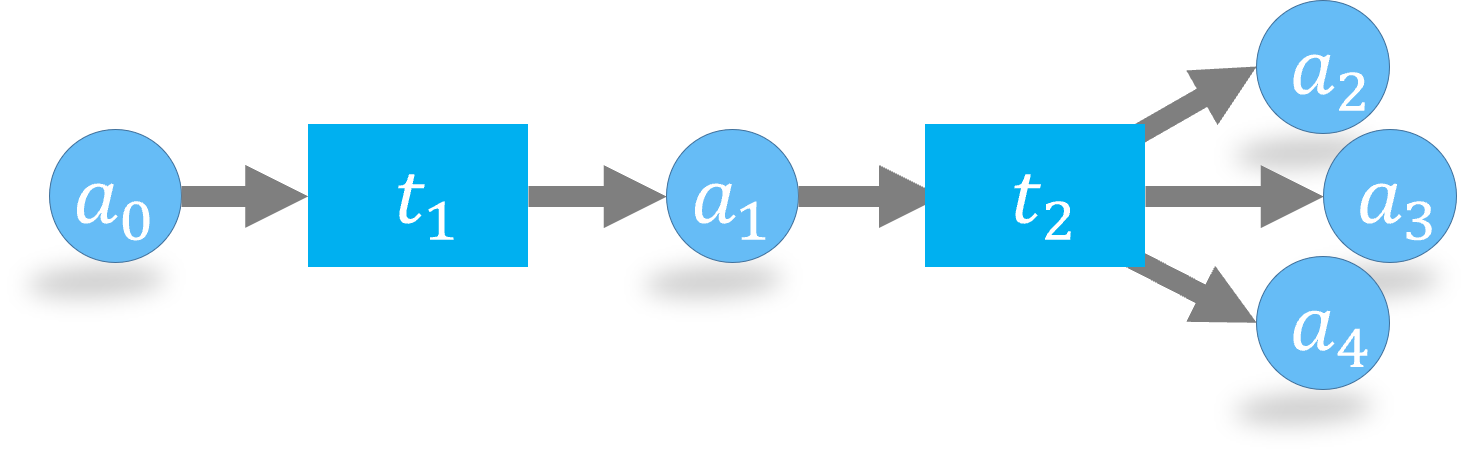}
    \caption{A 2-chainlet of five addresses. Each of the two transactions has one input address, but the $t_2$ has three output addresses.}
    \label{fig:twochainlet}
\end{figure}

\subsection*{Orbits}

Before we define chainlet orbits, we need to recall a few basic notions from graph theory. A mapping $\varphi: \G \to \h$ is called \textit{a graph isomorphism} if $\varphi: \V_\G\to \V_\h$ is a bijection such that $u,v\in \V_\G$ are adjacent in $\G$ if and only if $\varphi(u)$ and $\varphi(v)$ are adjacent in $\h$. If there exists an isomorphism from $\C_k$ to another chainlet $\C_k'$ of $\G$, we call $\C_k'$  an \textit{occurrence} of $\C_k$ 
in $\mathcal{G}$. We define the chainlet class of $\C_k$ in $\G$ as the set of all occurrences of $\C_k$ as follows:

\begin{definition} [Chainlet Classes] \label{def:class} Let $\A_k(\G)$ be the set of all $k$-chainlets in $\G$. For a given $k$-chainlet $\C_k\in \A_k(\G)$, the equivalence class of $\C_k$, $\E(\C_k)$, is the set of all isomorphic copies (occurrences) of $\C_k$ in $\A_k$. 
\end{definition}

While in the previous definition, we define equivalence classes of chainlets in $\G$, in the next definition we define an equivalence class of address nodes for a fixed chainlet class.

\begin{definition} [Orbit of an Address Node] \label{def:orbit} Let $\C_k$ be a $k$-chainlet in $\G$, and $u$ be an address node in $\C_k$. We define {\em the orbit of $u$ with respect to $\C_k$} as the collection of the images $\varphi(u)$ where $\varphi$ is an isomorphism from $\C_k$ to another $k$-chainlet $\C'_k$ in $\G$, i.e. $O_u=\{\varphi(u)\}$. We call $u\in \C_k$ and $v=\varphi(u)\in \C'_k$ {\em similar nodes} wrt $\C_k$, i.e., $u\sim v \ (\mbox{mod } C_k)$.  
\end{definition}

Notice that the similarity of two nodes depends on the fixed chainlet class $\E(\C_k)$. 
In other words, a node can belong to different chainlets $\C_k$ and $\C_k'$ because of their relative locations to the transaction nodes. Hence, the orbit of $u$ with respect to $\C_k$ versus with respect to $\C_k'$ can be completely different. In the remaining text, if the fixed chainlet type $\C_k$ is understood, we will refer to it as simply \textit{orbit of u}, $O_u$. Note also that if $u\sim v$, then the corresponding chainlets containing $u$ and $v$ must be in the same equivalence class of chainlets. In other words, any orbit $O_u$ belongs to only one chainlet class $\E(\C_k)$. Furthermore, by definition, if two nodes are similar, then they have the same orbit and vice versa, i.e. $u\sim v\iff O_u=O_v$.

The following theorem gives the relationship between the size of the orbit $O_u$ of a node $u\in \C_k$, and the size of the equivalence class of the chainlet $\E(\C_k)$.  Let $\I(\C_k)$ be the set of all isomorphisms between any two $k$-chainlets in $\E(\C_k)$. $\Lambda(\C_k)$ be the group of all automorphisms (self-isomorphisms) of $\C_k$. Then, for any node $u$ in $\C_k$, we define the stabilizer subgroup of $u$ in $\Lambda(\C_k)$ as $\s(u)=\{\varphi\in \Lambda(\C_k) |\varphi(u)=u\}$. Let $|K|$ represent the cardinality of the set $K$. Now, we are ready to state our theorem.

\begin{theorem} Let $\C_k$ be a $k$-chainlet in $\G$, and $u$ be a node in $\C_k$. Let $I(\C_k),\Lambda(\C_k), S(u)$ be as defined above. Then, we have $$|\I(\C_k)|=|\E(\C_k)|^2\cdot|\Lambda(\C_k)|  \quad \mbox{and} \quad |O_u|= |\E(C_k)|\cdot\frac{|\Lambda(\C_k)|}{|\s(u)|}$$
\end{theorem}

\begin{proof} We start with proving  the first equation $|\I(\C_k)|=|\E(\C_k)|^2\cdot|\Lambda(\C_k)|$. Let $|\E(\C_k)|=m$ and $\E(\C_k)=\{\C_k^1, \C_k^2, \dots, \C^m_k\}$ where $\C_k=\C_k^1$. For each $\C_k^i$, fix $\varphi_i\in \I(\C_k)$ such that $\varphi_i:\C_k^1\to\C_k^i$. Let $\psi\in \I(\C_k)$ be any isomorphism, i.e., $\psi:\C_k^{j_1}\to \C_k^{j_2}$. Then, $\widehat{\psi}=\varphi^{-1}_{j_2}\circ \psi \circ \varphi_{j_1}$  is an automorphism of $\C_k$, i.e. $\wh{\psi}:\C_k\to\C_k$. This implies $\wh{\psi}\in \Lambda(\C_k)$. Hence, by using fixed isomorphisms $\{\varphi_1, \dots,\varphi_m\}$, we find a unique automorphism $\wh{\psi}\in \Lambda(\C_k)$ for any isomorphism $\psi\in \I(\C_k)$. Notice that this process is reversible, i.e., for any $1\leq i\leq j\leq m$, and for any $\psi\in \Lambda(\C_k)$, we obtain a unique isomorphism $\wt{\psi}_{ij}=\varphi_j\circ\psi\circ\varphi_i^{-1}:\C^i_k\to\C_k^j$. This implies $|\I(\C_k)|=|\E(\C_k)|^2\cdot|\Lambda(\C_k)|$. The first equation follows.

To prove the second equation $|O_u|= |\E(C_k)|\cdot\frac{|\Lambda(\C_k)|}{|\s(u)|}$, we first show that $\frac{|\Lambda(\C_k)|}{|\s(u)|}$ is equal to the number of similar nodes of $u$ in $\C_k$, i.e. $\frac{|\Lambda(\C_k)|}{|\s(u)|}=|O_u^1|=|O_u\cap\V_{\C^1_k}|$ where $O^i_u$ represent the similar nodes to $u$ belonging to $\C^i_k$. Notice that $\Lambda(\C_k)$ is a group acting on the set $\C_k$. Hence, the size of the orbit  of $u$ in $\C^1_k$  is equal to the cardinality of the quotient of $\Lambda(\C_k)$ by the stabilizer subgroup $\s(u)$ by \cite[Proposition 5.1]{lang2012algebra}. In other words, $O_u^1$ has one to one correspondence with the quotient group $\Lambda(\C_k)/\s(u)$ (\# of cosets). Hence, if $O_u^1=\{u_1, u_2,\dots, u_r\}$, then  $\Lambda(\C_k)/\s(u)=\{\s(u)u_1, \dots, \s(u)u_r\}$, i.e., $|O_u^1|=\frac{|\Lambda(\C_k)|}{|\s(u)|}$

Now, take any node $v\in O_u$. Assume $v\in \C_k^i$ where $\C_k^i\in \E(\C_k)$. By definition of the whole orbit $O_u$, there exists an isomorphism $\psi:\C_k^1\to \C_k^i$ with $\psi(u)=v$. Since $\varphi_i^{-1}\circ \psi\in \Lambda(\C_k)$, $\varphi_i^{-1}(v)\in O_u^1$. This proves there is a one-to-one correspondence between the sets  $ O_u^1\times\{\varphi_1, \dots,\varphi_m\}$ and $O_u$. In other words, $O_u=\bigcup_{i=1}^m O_u^i$ where $O_u^i$ is the orbit of $u$ in $\C_k^i$. The proof follows.    
\end{proof}

\noindent\textbf{Choice on $k$-chainlets.} In this article, we focus on $2$-chainlet classes, and orbits of the addresses defined by these $2$-chainlets because this choice provides rich enough and effective structural topological embeddings that capture important transaction behavior. We chose to not utilize one-chainlets as they only model the receiving behavior of coins, whereas we wanted to model both the sending and receiving behavior.  $k$=2 is a natural choice that suits all application settings. Note that $k$=2 subsumes $k$=1 (as we show in only-active and only-passive experiments in \Cref{tab:allauc}). $k>3$ is not useful because once bitcoins are received (modeled with $k$=1) and spent (modeled with k=2), the owner of the coins has no control over what transactions the new buyer of the coins creates.  In this sense, $k$=2 is necessary and sufficient in understanding spending behavior. Beyond 2, we can only model the global behavior in the graph.  See appendix Figures~\ref{fig:orbits0}, \ref{fig:orbitsx1n}, \ref{fig:orbitsx2n}, \ref{fig:orbitsx2n}, \ref{fig:orbitsx3n} for all 48 orbits.

When counting the total number of distinct orbits, we exclude the input addresses of transaction $t_1$ (i.e., the first transaction involved in a $2$-chainlet) as their inclusion leads to a count of more than 120 orbits. Figure~\ref{fig:activepassive} shows an example of the resulting reduced chainlet form. By disregarding the input information of the first transaction, we reduce the orbit count to 48 (see Appendix for definitions and shapes). As an address typically has only 1.5 orbits on average, utilizing this reduced set of 48 orbits mitigates the sparsity issues in the extracted orbit data. Our experiments, presented in Section~\ref{sec:experiments}, demonstrate that our orbit definition provides strong classification power.

\noindent \textbf{Active and Passive Orbits:} We define the orbit of an address as active or passive based on the role of the address within a $2$-chainlet.   

\begin{definition} [Active Orbit] Let $\C$ be a $2$-chainlet. An active orbit with respect to $\C$ is an orbit where its addresses appear as the output of the first transaction in the corresponding occurrence of $\C$. When $\C$ is understood, we denote the set of active orbits as $O_{act}$.
\end{definition}
Active orbits indicate behavior created by the address owner, as the address in this orbit initiates a new transaction and sends coins to other addresses. We define  passive orbits  as follows:
\begin{definition} [Passive Orbit] Let $\C$ be a $2$-chainlet. A passive orbit with respect to $\C$ is an orbit where it appears as the output of the second transaction in the corresponding occurrence of $\C$. Again, when $\C$ is understood, we denote the set of passive orbits as $O_{pass}$.  
\end{definition}

Passive orbits for a chainlet class $\E(\C)$, on the other hand, receive coins without shaping any occurrences of $\C$, as the address in this orbit only receives coins from another address, without initiating a new transaction. It is worth noting that an address can belong to an active orbit for one chainlet class while belonging to a passive orbit for a different chainlet class. This means that an address may exhibit both initiating and receiving behavior in different $2$-chainlet types. Figure~\ref{fig:activepassive} shows the example of an active orbit 5 and a passive orbit 6 in a  2-chainlet.

\begin{figure}
    \centering
    \includegraphics[width=0.6\linewidth]{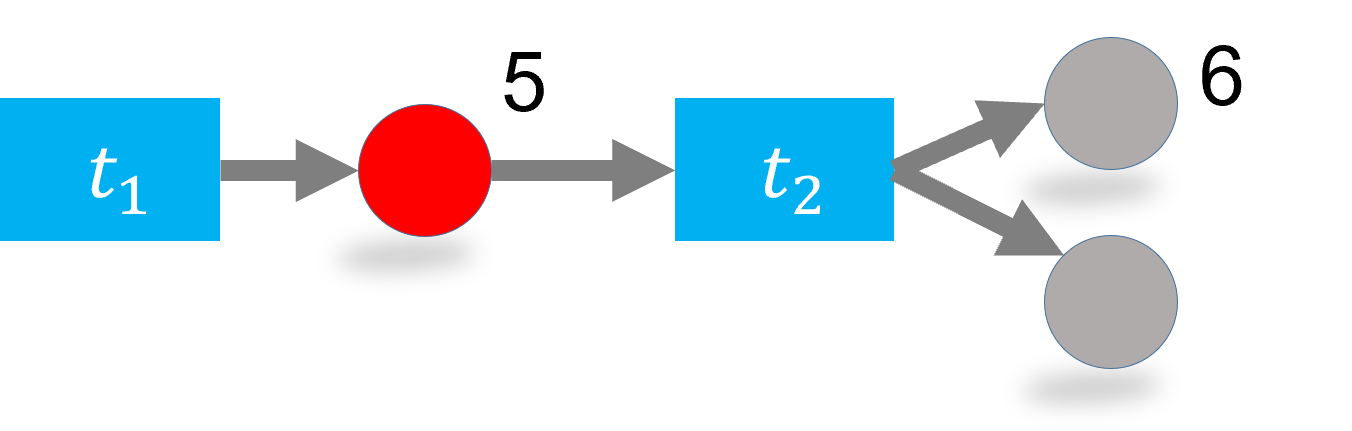}
    \caption{Examples of an active orbit 5 (in red) and a passive orbit 6 (in gray) in a 2-chainlet. Both outputs of $t_2$ are orbit 6 addresses. Address with orbit 5 initiates the second transaction, $t_2$, in the 2-chainlet, thus its spending behavior shapes the structure of the chainlet and determines the orbit type of the associated address.}
    \label{fig:activepassive}
\end{figure}

The distinction between active and passive orbits is significant as it allows us to focus on specific types of behavior within the transaction graph. Active orbits, for example, can be used to train a classifier for identifying spending patterns associated with coins gathered from darknet marketplaces and ransomware payments. This is because active orbits indicate behavior initiated by the address owner. On the other hand, passive orbits, which are shaped by payers, can be used to study the behavior of victims, as these orbits simply receive coins without shaping any chainlet. By understanding the behavior of both active and passive orbits, we can gain a more comprehensive understanding of the transaction graph and the actors within it.

\noindent\textbf{Computational Complexity of Orbit Discovery:}\label{sec:orbitcompl} Extracting a 2-chainlet requires three searches on the adjacency matrix. First, we must find all transaction nodes which has a cost of $\mathcal{O}(\V)$. Next, for every transaction pair $t_1$ and $t_2$, we must find if $ \exists a \in \{O(t_1) \cap I(t_2)\}$, where $O(t_1)$ represents the set of outputs of transaction $t_1$ and $I(t_2)$ represents the set of inputs of transaction $t_2$. The second step thus has a complexity of $\mathcal{O}(\V^3)$. However, 57.40\% of Bitcoin transactions have less than two input and two output addresses. Hence the complexity is reduced to $k\times \mathcal{O}(\V^2)$  where $k \approx 2$. As a result, finding all address orbits of a daily Bitcoin graph (approximately $400K$ transactions) takes around 15 minutes on an off-the-shelf computer. Please refer to \Cref{fig:time} for the visualization of the orbit extraction time in the daily heterogeneous Bitcoin transaction graph.

\section{Experiments}
\label{sec:experiments}

\noindent\textbf{Datasets}
We extract data from the Bitcoin transaction network and use labels for ransomware and darknet market addresses that have been obtained from external sources. 

\noindent\textbf{Bitcoin Transaction Network.} We have downloaded and analyzed the complete Bitcoin transaction graph from its inception in 2009 until 2023 by utilizing the Bitcoin Core Wallet and the Bitcoin-ETL library~\url{https://github.com/blockchain-etl/bitcoin-etl}. By implementing a 24-hour time window centered on the Central Standard Time zone, we extract daily transactions on the network and construct the Bitcoin graph.

\noindent\textbf{Ransomware Addresses.} Our ransomware dataset is a union of datasets from three widely adopted studies:    
Montreal~\cite{paquet2018ransomware}, Princeton~\cite{huang2018tracking} and  Padua~\cite{conti2018economic}. The combined dataset contains addresses from 27 ransomware families. 19930 unique ransomware addresses appear a total of 62641 times between 2009-2023.

\noindent\textbf{Darknet Addresses.} We have downloaded the Grams dataset from Darknet Market Archives~\cite{gwerndarknet}. Grams was a platform that primarily focused on facilitating searches for market listings. The platform utilized a variety of methods to acquire the listings, including utilizing API exports provided by various markets. The dataset spans from June 9, 2014, to July 12, 2015. We identify 7557 unique addresses associated with darknet marketplaces a total of 1288100 times.

\subsection{Experimental Setup}
\noindent{\bf Hardware:}  We ran experiments on a single machine with a Dell PowerEdge R630, featuring an Intel Xeon E5-2650 v3 Processor (10-cores, 2.30 GHz, 20MB Cache), and 192GB of RAM (DDR4-2133MHz).

\noindent{\bf Preprocessing:} We do not apply any normalization or preprocessing to the address orbits. Because extracting orbits is not demanding in terms of computational resources, we were able to use the entire transaction network without removing any edges or nodes. The address orbits, which are 27 GB in size, are publicly available at \url{***.org/btc/orbits}. 

Orbit discovery is an efficient process on the large Bitcoin transaction network. As \Cref{fig:time} shows, on the majority of days, orbit extraction is completed within 1000 seconds (equivalent to 15 minutes).
\begin{figure}
    \centering
    \includegraphics[width=0.8\linewidth]{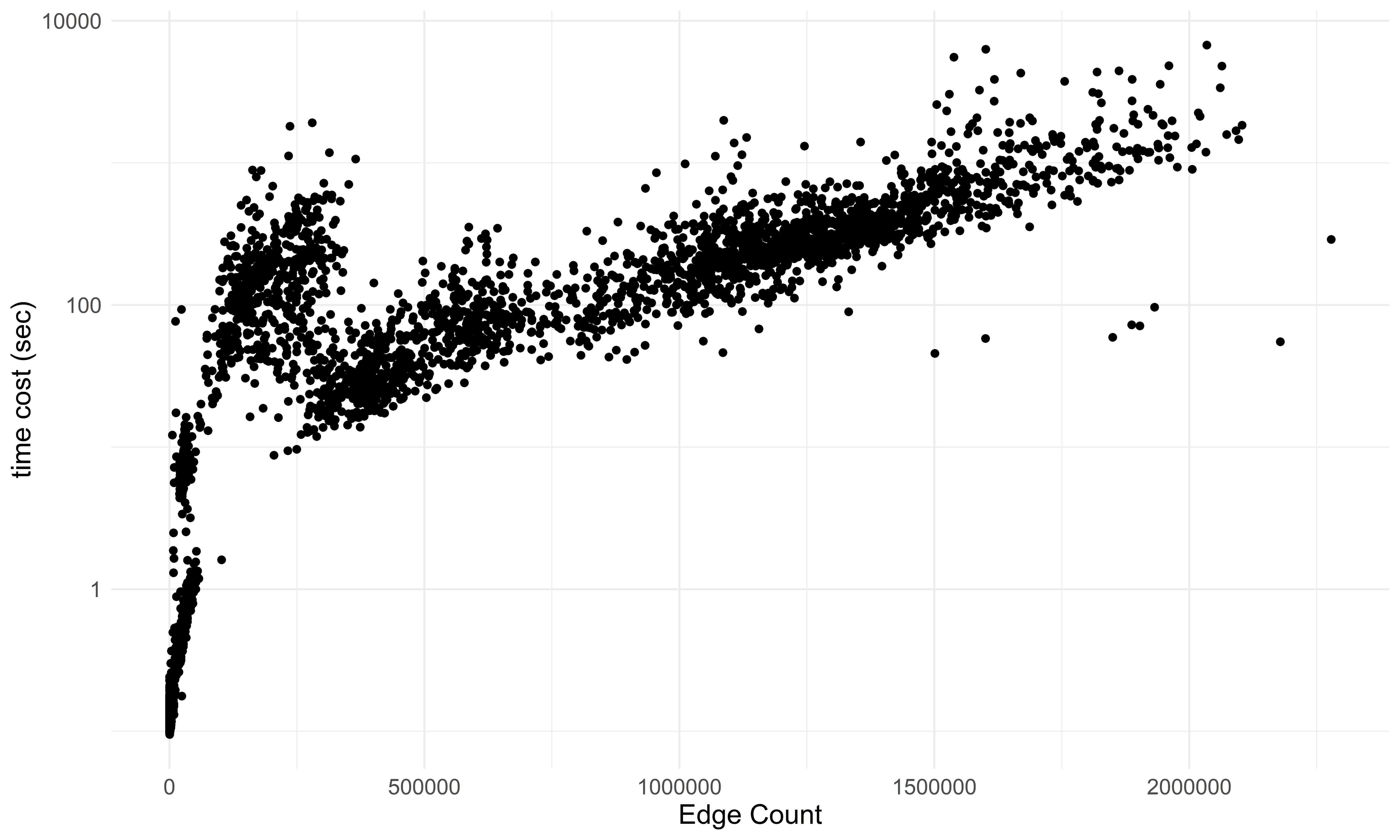}
    \caption{Time required to extract orbits from the daily Bitcoin transaction graph, which consists of two node types (addresses and transactions) and their edges. Each data point represents the time costs associated with extracting orbits for a particular day.}
    \label{fig:time}
\end{figure}

\noindent{\bf Code and Parameters:} Our code and details on hyperparameter selection and models are given at  \url{https://github.com/chainletRepo/chainlet}.

\subsection{Results}
This section begins by presenting the outcomes of widely employed heuristics for clustering Bitcoin addresses. Then, we outline the tactics utilized in e-crime before proceeding to the primary findings of this article, including the BitcoinHeist benchmark and the full orbit dataset.

\noindent\textbf{Heuristics.} By employing the co-spending and transition heuristics (\Cref{sec:prelim}) with all history of ransomware addresses, we discover only 40 unique addresses from  CryptoLocker (Padua), CryptoWall (Padua), CryptoTorLocker2015 (Montreal), and CryptoTorLocker2015 (Padua) families. These low numbers are also because researchers of the datasets had already used the heuristics to expand the datasets. As a result, we do not find many more undisclosed addresses.

\subsubsection{Adversarial Patterns of E-Crime}

\input{tables/tableAll.tex}

Our analysis has revealed a total of 22682 unique orbit patterns associated with addresses linked to ransomware and 546838 unique orbit patterns associated with addresses associated with darknet marketplaces. Typically, an address is only connected to a single occurrence within an orbit for any given day, resulting in an orbit value of 1. \Cref{tab:freqorbits} lists the most frequent orbit patterns for classes. 

Additionally, 2.63\% of all darknet market addresses appear in two 2-chainlets daily, within orbits $\ob13$ and $\ob14$. However, it should be noted that these orbits are prevalent across all types of addresses, including those associated with ransomware. As \Cref{tab:freqorbits} illustrates, the top three patterns associated with ransomware addresses make up 16.53\% of all ransomware addresses, indicating a lack of diversity in ransomware address behavior. This implies that ransomware operators tend to use a limited number of patterns, which makes tracking ransomware payments easier. 

\input{tables/tableRS.tex}

 \Cref{tab:rsFreq} further delves into the RS patterns. In the table, we binarize orbit counts and use a checkmark to indicate that an address has a non-zero count for the associated orbit. Certain patterns exhibit a strong association with addresses linked to ransomware. A significant proportion of addresses (95.9\%) with the pattern in the first row of the table are identified as ransomware addresses. It is highly probable that the remaining white address with this pattern is also associated with e-crime activities.  \Cref{tab:allFreqPatterns} shows the most and the least likely patterns associated with white addresses. 

\input{tables/tableAllFreq.tex}

Our analysis has provided us with the following insights into e-crime behavior.  
\begin{itemize}[leftmargin=*]
\item On average, white addresses have three distinct non-zero orbits, whereas darknet addresses have six non-zero orbits. Similarly, RS addresses have six non-zero orbits. These values are higher than expected because the Bitcoin community discourages address reuse in transactions. As a result, we would expect most addresses to participate in fewer chainlets. However, the high values are likely a result of multiple payments being made to the same address, which indicates a lack of common privacy considerations in these types of e-crime transactions.
    \item Orbit 1 in darknet addresses reveals that coins transferred in a transaction are not re-transferred on the same day, indicating that darknet payments are retained in the receiving address. This aligns with the previous escrow practice of darknet markets where payments to sellers were withheld until buyers confirmed receipt of goods. Orbit 1 does not appear in the top 10 most frequent orbits of either ransomware or white addresses.
    \item Orbits 9 and 12 in ransomware addresses indicate immediate forwarding of payments to new addresses (see \Cref{fig:orbitsx2n}). 
     
    \item Orbit 9 indicates merges coins from multiple payments. However, none of these orbits connect to coin-mixing transactions (see the appendix for coin-mixing transactions). 
    \item Orbits 30, 31, 32, 39, 40, 41, 46, and 47 are observed in chainlets that may be involved in coin-mixing transactions. Interestingly, these orbits are not present in the top 10 most common patterns for any address type. E-crime coins are first combined in an address and then coin-mixed. 
\end{itemize}

\subsubsection{BitcoinHeist Benchmark}
 
We use orbits of Bitcoin addresses in a binary (white vs. ransomware) address classification task on the BitcoinHeist dataset and compare our results to those reported in the article by Akcora et al.~\cite{akcora2021bitcoinheist}. We obtained 2916697 ransomware and white addresses from the authors of the dataset. The BitcoinHeist dataset comprises 1048576 transactions from January 2009 to December 2018. Transactions with amounts less than $\bitcoin 0.3$ were removed as ransomware payments are usually significantly larger. Since the dataset only includes ransomware and white addresses, we present the performance of our methods in terms of binary classification by using a Random Forest~\cite{breiman2001random} with 300 trees. 

In 24 ransomware families, at least one address appears in more than one 24-hour time window.   CryptoLocker has  13 addresses that appear more than 100 times each. The CryptoLocker address {\tiny 1LXrSb67EaH1LGc6d6kWHq8rgv4ZBQAcpU} appears for a maximum of 420 times.  Four addresses have conflicting ransomware labels between Montreal and Padua datasets. APT  (Montreal) and Jigsaw (Padua) ransomware families have two and one \textit{multisig} addresses, respectively. All other addresses are ordinary addresses that start with '1'.

\begin{table}[t]
\caption{AUC scores for address classification using features from the BitcoinHeist dataset and Orbits. White addresses were under-sampled without replacement, while ransomware addresses were not under-sampled. The results indicate that using Orbits improves address classification significantly.\label{tab:sampleHeist}}
    \centering
    \begin{tabular}{l c c c}
    \toprule
        AUC@Sample&	BitcoinHeist&	Orbit&	BitcoinHeist $+$ Orbit\\
        \midrule
AUC@1M&	 0.791(0.005)&	 0.839(0.003)&	0.895(0.006)\\
AUC@500K&	0.798(0.006)&	0.849(0.002)&	0.903(0.003)\\
AUC@100K&	0.808(0.003)&	0.856(0.005)&	0.908(0.003)\\
AUC@50K	&0.807(0.006)&	0.857(0.004)&	0.908(0.003)\\
\bottomrule
    \end{tabular}
\end{table}

\Cref{tab:sampleHeist} shows the AUC scores of the ROC curve for the binary address classification task. In BitcoinHeist+Orbit, we combine feature sets of addresses. The AUC of the BitcoinHeist and BitcoinHeist+Orbit models improves as the white sample rate increases from 50K to 100K, but larger sample rates decrease performance. However, the performance of the Orbit model remains relatively stable with larger white sampling rates, decreasing only by 0.018 from 50K to 1M. This suggests that orbits are relatively robust against the class imbalance issues that often occur in blockchain data analytics, as labeled data is often scarce in these tasks. 

The superiority of the BitcoinHeist+Orbit model over the Orbit model can be attributed to the integration of the income feature from the BitcoinHeist dataset, which calculates the sum of incoming coin amounts for each address. Please note that orbits do not consider the coin amounts that are transferred. As the ransom amounts are typically consistent (e.g., $\bitcoin$0.5), the income provides a good feature to predict the ransomware label. The Mean Decrease in Impurity~\cite{breiman2001random} identifies income, $\ob9$, and $\ob12$ as the three most important features of the classifier. Incorporating the income feature into the orbits data is computationally efficient (i.e., $O(\V)$) as it only requires summing the incoming amounts for each address.

It's worth mentioning that the BitcoinHeist study employs graph flow-based features for each address, which are computationally intensive. To minimize computational expenses, the authors applied edge filtering with a threshold of 0.3 Bitcoins. On the other hand, orbits are not as computationally demanding (\Cref{sec:orbitcompl}), thus, we were not constrained by the computational complexity issues. This permitted us to extract orbits for all Bitcoin addresses within the examined period.

\begin{table}[]
\caption{AUC (and standard deviation in parentheses) values for address classification using BitcoinHeist features and Orbits. Both white (W) and ransomware (RS) addresses are under-sampled without replacement. }
    \label{tab:gnncomparison}
    \centering
    \begin{tabular}{r l l l l l}
    \toprule
        Sample&	 &	GCN& GIN&DGCNN  	\\
        \midrule
20K&	 &	0.76(0.002) &	0.90(0.02)&0.79(0.09)  \\
10K&	 &	0.70(0.004) &0.88(0.01)	& 0.75(0.09) \\
5K&	 &	0.70(0.002) &0.86(0.02)	& 0.78(0.09) \\
\bottomrule
    \end{tabular}
\end{table}

\noindent\textbf{Graph Neural Networks on the BitcoinHeist Benchmark.} We evaluate the orbits-based method against four graph classification techniques on the BitcoinHeist data utilizing graph neural networks: a GCN~\cite{kipf2016semi}, DGCNN~\cite{zhang2018end}, GIN~\cite{xu2018powerful} and GraphSage~\cite{hamilton2017inductive}. Our objective is to evaluate the performance of orbit-based classifiers by comparing them to Graph Neural Network (GNN) classifiers. While there are numerous methods, we are not conducting a comprehensive comparison of all of them. Instead, we aim to provide a good understanding of what GNNs can achieve. We report the average accuracy of five runs, along with their standard deviation.  In these models,  we adopted the grid search parameters from the benchmark article by Errica et al.~\cite{erricafair20}. 

We could not apply existing node (i.e., address) classification approaches~\cite{xiao2022graph} to the Bitcoin transaction network because the Bitcoin transaction network is too large and contains as few as four e-crime addresses per day (over ten thousand days).  Instead, we re-formulated the address classification task as the classification of the subgraph of the two-hop neighborhood of an address. Formally, the neighborhood graph of an address is the combination of the induced subgraphs of i) $2$-chainlets where the second transaction outputs coins to the address and ii) $2$-chainlets where the first transaction outputs coins to the address. An example of a neighborhood graph is shown in \Cref{fig:rsmostfreq}. 

There are several advantages of using the neighborhood approach instead of a traditional node classification approach. Firstly, the neighborhood approach enables an inductive learner where an address from any day can be classified by using its neighborhood graph from that day, while transductive node classification is restricted to the specific graph on which it was trained and cannot be applied to new, unseen graphs. Secondly, the neighborhood approach provides us with the ability to control the number of addresses in the model, as GNNs can have enormous computational costs when applied to very large datasets.

Despite utilizing the neighborhood graph, neural networks still face the challenge of data imbalance. In the case of more than 600K addresses daily, only a small number of e-crime addresses exist. To address this problem, we take two steps. Firstly, we create a balanced dataset to evaluate the graph classification methods where close to half of the addresses are white (i.e., 57\%). Secondly, we restrict the number of addresses fed into the model. 

The resulting AUC scores are shown in \cref{tab:gnncomparison}. In addition to the models mentioned, we attempted to train a robust GraphSage~\cite{xu2018powerful} model, but it did not perform well as the model AUC did not improve throughout the training. We believe that this could be due to the sparsity of the Bitcoin transaction network, where many addresses have only one or two neighbors, and the usage of sampling-based approaches leads to disconnected graphs.

\begin{figure*}[t]
\begin{subfigure}{.31\textwidth}
  \centering 
  \includegraphics[height=0.85\linewidth]{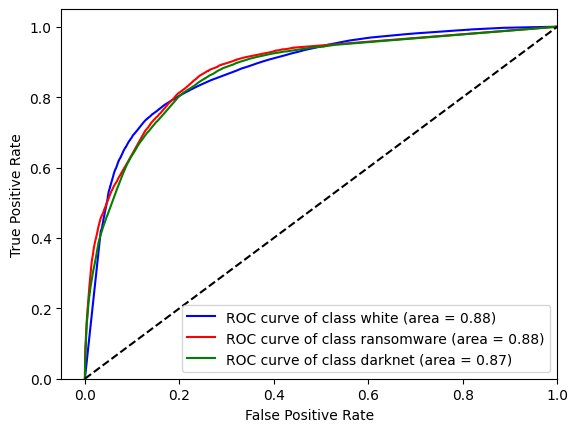} 
  \caption{\footnotesize  Sample=100K.}
  \label{fig:auc100k}
\end{subfigure}
~
\begin{subfigure}{.31\textwidth}
  \centering 
  \includegraphics[height=0.85\linewidth]{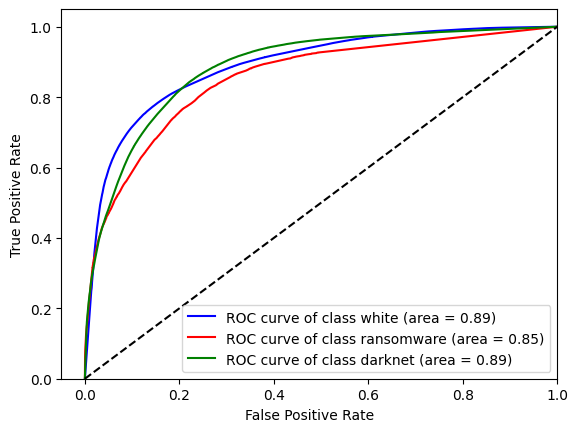} 
  \caption{\footnotesize  Sample=500K.}
  \label{fig:auc500k}
\end{subfigure}
~
\begin{subfigure}{.31\textwidth}
  \centering 
  \includegraphics[height=0.85\linewidth]{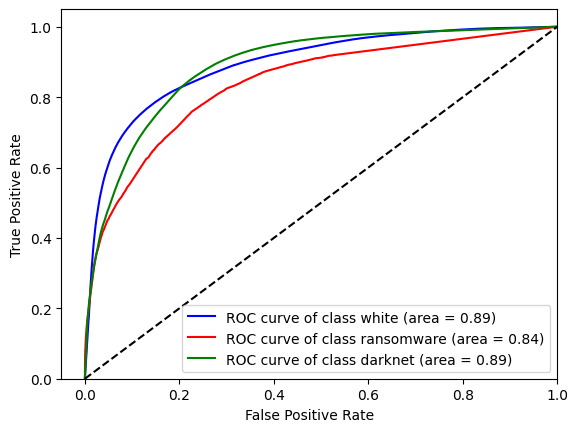} 
  \caption{\footnotesize  Sample=1M.}
  \label{fig:auc1m}
\end{subfigure}
 \caption{AUC scores for Orbit-based e-crime address classification in the period of 2009-2023. White and darknet classes are undersampled. Orbit performance is quite robust against the increasing sample sizes.} 
\label{fig:allorbit}
\end{figure*}

As shown in \Cref{tab:gnncomparison}, GIN is a powerful model and consistently improves its AUC score with more training data. However, the complexity of the model limits the amount of training data that can be used. However, these results are achieved over balanced data assumptions which do not hold over all data. As we show in the next section, our orbit method is capable of achieving improved AUC scores for the challenging task of three-way classification when trained on large e-crime data.

\noindent\textbf{Orbits vs. GIN.} The results of GIN in \Cref{tab:gnncomparison} (0.9 AUC) are similar to those of the Orbit results in \Cref{tab:sampleHeist} (0.857 AUC). However, GIN is limited to smaller datasets of under 20K addresses, while Orbit-based models can utilize a full day's worth of addresses (over 500K) to make predictions for the next day. This allows us to incorporate temporal patterns into our models and give more importance to recent Orbit patterns.

\subsubsection{Orbits on E-crime Graph Classification}

The goal of this experiment is to generalize the orbit e-crime address classification performance with a larger dataset.  We take a sample of 1\% of the daily white addresses in the Bitcoin history, resulting in 3.16 million addresses, and also include all  1.28 million darknet market and 64,000 ransomware addresses.

\begin{table}[]
\caption{AUC scores for active ($O_{act}$) and passive ($O_{pass}$) orbits on 1\% of Bitcoin addresses (4.3M addresses).}
    \label{tab:allauc}
    \centering
    \begin{tabular}{c c c}
    \toprule 
    Class&$O_{act}$ AUC&$O_{pass}$ AUC\\
    \midrule
        White &0.731(0.006) &0.870(0.002)  \\ RS &0.725(0.004) &0.736(0.001)  \\ DM &0.727(0.004) &0.876(0.001)  \\
        \bottomrule
    \end{tabular}
\end{table}

\Cref{fig:allorbit} illustrates the three-way address classification performance for white, ransomware, and darknet addresses. The figure tracks the AUC score for each class with a one-versus-rest Random Forest classifier, as more white data is added in three steps to the classifier input. Training/test split is 80/20. The consistently high AUC score, above 0.8, demonstrates the strength of the classifier.

Furthermore, we investigate the role of different types of orbits in detecting e-crime payments. We consider two types of orbits as outlined in \Cref{sec:chainlet}: active orbits, which refer to the behavior of e-crime operators such as darknet market sellers and ransomware hackers, and passive orbits, which refer to the behavior of darknet market buyers or ransomware victims.  In \Cref{tab:allauc}, we evaluate the performance of different feature sets by limiting the data features to either active or passive orbits and calculating the macro-averaged AUC performance.

The results show that reducing the feature set to active orbits lead to a decrease in AUC scores compared to using both feature sets, which has AUC scores of 0.84 or higher (\Cref{fig:auc1m}).

However, the results also indicate that the passive orbits yield better AUC results, meaning that the behavior of victims or buyers is more informative in detecting e-crime payments. Furthermore, active orbit  AUC values per class are similar, which suggests that the performance of the classifier is consistent across different classes.

In conclusion, these results highlight the importance of considering both active and passive orbits in detecting e-crime payments. While passive orbits provide valuable information, the best classification performance is achieved by utilizing both active and passive orbits together.

\section{Conclusions}
\label{sec:conc}

 We propose a topological node embedding model, called "orbit," that captures a node's structural role in a UTXO graph. This approach results in compact and effective embeddings that can be easily integrated into node and graph classification models which outperforms current state-of-the-art methods in blockchain graph node classification tasks. Additionally, the orbit patterns can be visualized and used with interpretable and explainable models.

Our approach is not limited to Bitcoin graphs only and can be applied to UTXO graphs of other blockchains, such as ZCash and Monero, where e-crime is a prevalent issue. The orbit discovery process is computationally efficient and can be applied to very large graphs easily.

Future work could focus on  multi-modal graphs where meta attributes of transactions can be incorporated as node features. Another promising area of research is developing self-supervised ways to learn node and edge functions from data to define new orbits.

\section{Acknowledgements}
This work is supported in part by National Science Foundation under Grant No. ECCS 2039701, ECCS 1824716, DMS 1925346, CNS 1837627, OAC 1828467, IIS 1939728, CNS 2029661 and Canadian NSERC Discovery Grant RGPIN-2020-05665.

\bibliographystyle{plain}
\bibliography{orbits.bib}

\newpage
\input{sections/apparxiv}

\end{document}

%% file: tables/tableAll.tex
\begin{table}
\caption{Top-3 most frequent orbit counts of darknet market, ransomware (RS) and white addresses. Orbits that do not appear in the top-3 have been  excluded. Percentages are computed within each address class. }
\label{tab:freqorbits}
\centering
\tiny
\begin{tabular}{l c c c c c c c}

\toprule
    	&\textbf{Rank}	&\textbf{$\ob1$}	&\textbf{$\ob9$}	&\textbf{$\ob12$}	&\textbf{$\ob13$}	&\textbf{$\ob14$}	&\textbf{Perc.}\\
\cmidrule{2-8}
		&1	&\textcolor{gray}{0}	&\textcolor{gray}{0}	&\textcolor{gray}{0}	&\textcolor{red}{1}	&\textcolor{red}{\textbf{1}}	&2.63\\

Darknet	&2	&\textcolor{red}{\textbf{1}}	&\textcolor{gray}{0}	&\textcolor{gray}{0}	&\textcolor{gray}{0}	&\textcolor{gray}{0}	&1.13\\

		&3	&\textcolor{red}{\textbf{1}}	&\textcolor{gray}{0}	&\textcolor{gray}{0}	&\textcolor{gray}{0}	&\textcolor{red}{1}	&0.97\\
\cline{2-8}
		&1	&\textcolor{gray}{0}	&\textcolor{red}{1}	&\textcolor{gray}{\textbf{0}}	&\textcolor{red}{1}	&\textcolor{red}{\textbf{1}}	&6.69\\

RS		&2	&\textcolor{gray}{0}	&\textcolor{gray}{0}	&\textcolor{gray}{0}	&\textcolor{red}{1}	&\textcolor{red}{\textbf{1}}	&6.61\\

		&3	&\textcolor{gray}{0}	&\textcolor{gray}{0}	&\textcolor{red}{\textbf{1}}	&\textcolor{red}{\textbf{1}}	&\textcolor{red}{1}	&3.23\\
\cline{2-8}
		&1	&\textcolor{gray}{0}	&\textcolor{gray}{0}	&\textcolor{red}{\textbf{1}}	&\textcolor{gray}{0}	&\textcolor{red}{1}	&5.93\\

White	&2	&\textcolor{gray}{0}	&\textcolor{gray}{\textbf{0}}	&\textcolor{red}{\textbf{1}}	&\textcolor{red}{\textbf{1}}	&\textcolor{red}{1}	&4.89\\

		&3	&\textcolor{red}{\textbf{1}}	&\textcolor{gray}{0}	&\textcolor{gray}{0}	&\textcolor{gray}{0}	&\textcolor{gray}{0}	&4.40\\

\bottomrule
\end{tabular}
\end{table}

%% file: tables/tableRS.tex
 \begin{table*}[]
 \caption{Top-5 orbit patterns where addresses have the highest likelihood of being associated with ransomware. A checkmark indicates non-zero orbit count. The RS\% column in the table represents the percentage of addresses that are identified as ransomware addresses. W (white), DM (darknet market) and RS (ransomware) are counts of addresses with the related pattern.  For example, for the first orbit pattern, 95.9\% (i.e., 94) of the addresses with that pattern are identified as ransomware addresses, while the remaining addresses are either white or associated with darknet }
    \label{tab:rsFreq}
    \centering
    \tiny
    \begin{tabular} {p{0.10cm} p{0.10cm} p{0.10cm} p{0.10cm} p{0.20cm} p{0.20cm} p{0.20cm} p{0.20cm} p{0.20cm} p{0.20cm} p{0.20cm} p{0.20cm} p{0.20cm} p{0.20cm} p{0.20cm} p{0.20cm} p{0.20cm} p{0.20cm} r r r r}
         \toprule          $\ob1$&$\ob6$&$\ob8$&$\ob9$&$\ob13$&$\ob14$&$\ob17$&$\ob25$&$\ob28$&$\ob31$&$\ob32$&$\ob34$&$\ob37$&$\ob40$&$\ob42$&$\ob43$&$\ob45$&$\ob47$&W&DM&RS&RS\%\\
         \midrule 
 & &\checkmark& & & & &\checkmark&\checkmark&\checkmark& & &\checkmark& &\checkmark& &\checkmark& &1&3&94&95.9\\
 & &\checkmark& & & &\checkmark&\checkmark&\checkmark&\checkmark&\checkmark&\checkmark&\checkmark&\checkmark&\checkmark&\checkmark&\checkmark&\checkmark&64&3&216&76.3\\
\checkmark&\checkmark& &\checkmark&\checkmark&\checkmark& & & & & & & & & & & & &21&91&81&41.9\\
\checkmark&  &  &\checkmark&\checkmark&\checkmark&  &  &  &  &  &  &  &  &  &  &  &  &230&186&131&23.9\\
\checkmark&  &  &\checkmark&\checkmark&\checkmark&  &  &  &  &  &  &  &  &  &  &  &  &106&150&108&29.7\\
    \end{tabular}
\end{table*}

%% file: tables/tableAllFreq.tex
\begin{table*}[t] 
    \caption{The top-3 and bottom-3 orbit patterns with the lowest and highest likelihood of being white addresses, respectively. A checkmark  indicates  non-zero orbit count. The Non-White\% column in the table represents the percentage of addresses that are NOT identified as white addresses. W (white), DM (darknet market) and RS (ransomware) are counts of addresses with the related pattern.  For example, for the first orbit pattern, 99.9\% (i.e., 9800) of the addresses with that pattern are identified as ransomware, while the remaining addresses are white.}
    \label{tab:allFreqPatterns}
    \centering
    \tiny
        \begin{tabular}{p{0.10cm} p{0.10cm} p{0.15cm} p{0.15cm} p{0.15cm} p{0.15cm} p{0.15cm} p{0.15cm} p{0.15cm} p{0.15cm} p{0.15cm} p{0.15cm} p{0.15cm} p{0.15cm} p{0.15cm} p{0.15cm} p{0.15cm} p{0.15cm} r r r r r}
\toprule
 $\ob10$&$\ob12$&$\ob13$&$\ob14$&$\ob25$&$\ob28$&$\ob30$&$\ob31$&$\ob32$&$\ob34$&$\ob37$&$\ob40$&$\ob41$&$\ob43$&$\ob44$&$\ob45$&$\ob46$&$\ob47$&W&DM&RS&Non-White\%\\
\midrule
  & & & & & & &$\checkmark$&$\checkmark$& & &$\checkmark$&$\checkmark$& & & &$\checkmark$&$\checkmark$&4& 0&9800&99.9\\
  & & & & &$\checkmark$&$\checkmark$& &$\checkmark$& & & & & & & & & &61&14&7316&99.0\\
 $\checkmark$&$\checkmark$& & & & & & & & & & & & & & & & &116&10&10666&98.8\\
\midrule
  & & & &$\checkmark$&$\checkmark$& &$\checkmark$&$\checkmark$&$\checkmark$&$\checkmark$& & &$\checkmark$&$\checkmark$&$\checkmark$& &$\checkmark$&12915&0 &1249&8.8\\
 $\checkmark$&$\checkmark$&$\checkmark$&$\checkmark$& & & & & & & & & & & & & & &6029&1&442&6.8\\
 & & & &$\checkmark$&$\checkmark$& &$\checkmark$&$\checkmark$&$\checkmark$&$\checkmark$& & & & & & & &9309&0&533&5.4\\
\bottomrule
    \end{tabular}
\end{table*}

%% file: sections/apparxiv.tex
\section{Appendix}

In this part, we give further details of our method and experiments. In \Cref{sec:orbits}, chainlet orbits are visualized and described. In \Cref{sec:applic}, we describe how orbits can be applied to other UTXO blockchains. In \Cref{sec:mixing}, we explain the coin-mixing behavior and give our related orbits.

\section{UTXO network and Chainlets}
\label{sec:orbits}
In the context of the Bitcoin blockchain, a transaction can be formally defined in terms of its input and output as follows:
Input: A set of one or more digital signatures, or "unspent transaction outputs" (UTXOs), that authorize the transfer of value from the sender's Bitcoin address to the recipient's address. UTXOs are previously-created outputs of transactions that were not spent yet and they are used as an input in a new transaction.
Output: A set of one or more new UTXOs that represent the transfer of value from the sender's address to the recipient's address.
In other words, a Bitcoin transaction takes one or more UTXOs as input and creates one or more new UTXOs as output. The total value of the inputs must be greater than or equal to the total value of the outputs, otherwise the transaction will be considered invalid.

Figures~\ref{fig:orbits0}, \ref{fig:orbitsx1n}, \ref{fig:orbitsx2n}, \ref{fig:orbitsx3n} list all orbits and define address roles in them. We use the X-M-N notations to describe chainlets where the first chainlet has X inputs, M outputs whereas the second chainlet has N outputs. In orbits, we are particularly interested in M and N.

\section{Applicability to Monero and Zcash}
\label{sec:applic}

The concept of orbits can be applied to heterogeneous graphs like the UTXO transaction networks that consist of transaction and address nodes. In the case of Monero and ZCash, which both utilize the UTXO data model~\cite{akcora2022blockchain}, orbits can be directly applied to the unshielded transactions pool of ZCash without any modifications.

In Monero, the association between addresses and transactions is concealed in all transactions, but the input and output counts are not hidden. With orbits defined for both input and output addresses of transactions, these counts can be tracked. However, instead of representing addresses as nodes, output hash IDs would be represented as nodes in this context.

\begin{figure*}[b]
\begin{minipage}[b]{0.30\linewidth}
  \centering 
  \includegraphics[width=0.8\linewidth]{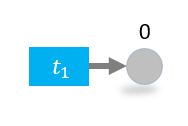} 
  \caption*{(a) Orbit   0:
if $\left|\Gamma_1^{o}\right|=1$ and $a \in \Gamma_1^{o}$ then $ a\in  \ob_0.$}
\end{minipage}
\hfill%
\begin{minipage}[b]{0.30\linewidth}
  \centering 
  \includegraphics[width=0.7\linewidth]{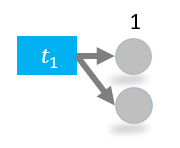} 
  \caption*{(b)  Orbit   1:
if  $\left|\Gamma_1^{o}\right|=2$  and $a \in \Gamma_1^{o}$ then $a \in \ob_1.$}
\end{minipage}
\hfill%
\begin{minipage}[b]{0.30\linewidth}
  \centering 
  \includegraphics[width=0.6\linewidth]{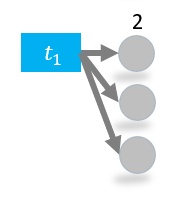} 
  \caption*{(c) Orbit 2: if $\left|\Gamma_1^{o}\right|=3$  and $a \in \Gamma_1^{o}$ then $a \in \ob_2.$}
\end{minipage}
\caption{Orbits 0 to 2 involve addresses who do not create a 2-chainlet for the daily network snapshot. These addresses keep the received coins dormant for the given day. All addresses have passive orbits.\label{fig:orbits0}} 
\end{figure*}

\begin{figure*}[]
\begin{minipage}[b]{0.30\linewidth}
  \centering 
  \includegraphics[width=\ra\linewidth]{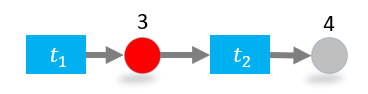} 
  \caption*{(d) Orbit   3:
if $\left|\Gamma_1^{o}\right|=1$, $\left|\Gamma_2^{i}\right|=1$, $\left|\Gamma_2^{o}\right|=1$ and  $a \in \Gamma_1^{o} $, $a \in \Gamma_2^{i} $ then $a \in \ob_3 $. Orbit   4:
if $\left|\Gamma_1^{o}\right|=1$, $\left|\Gamma_2^{i}\right|=1$, $\left|\Gamma_2^{o}\right|=1$ and $a \in \Gamma_2^{o}$, then $a \in \ob_4$. }
\end{minipage}
\hfill%
\begin{minipage}[b]{0.30\linewidth}
  \centering 
  \includegraphics[width=\ra\linewidth]{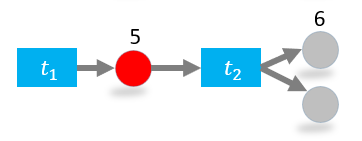} 
  \caption*{(e) Orbit 5:
if $\left|\Gamma_1^{o}\right|=1$, $\left|\Gamma_2^{i}\right|=1$, $\left|\Gamma_2^{o}\right|=2$, and $a \in \Gamma_1^{o} $, $a \in \Gamma_2^{i}$ then $a \in \ob_5$. 
 Orbit 6:
if $\left|\Gamma_1^{o}\right|=1$, $\left|\Gamma_2^{i}\right|=1$, $\left|\Gamma_2^{o}\right|=2$, and $a \in \Gamma_2^{o}$ then $a \in \ob_6.$ }
\end{minipage}
\hfill%
\begin{minipage}[b]{0.30\linewidth}
  \centering 
  \includegraphics[width=\ra\linewidth]{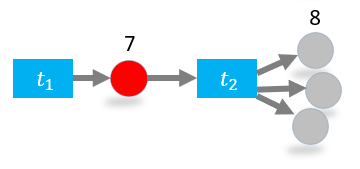} 
  \caption*{(f) Orbit 7:
if $\left|\Gamma_1^{o}\right|=1$, $\left|\Gamma_2^{i}\right|=1$, $\left|\Gamma_2^{o}\right|=3$, and $a \in \Gamma_1^{o} $, $a \in \Gamma_2^{i}$ then $a \in \ob7$.
 Orbit 8:
if $\left|\Gamma_1^{o}\right|=1$, $\left|\Gamma_2^{i}\right|=1$, $\left|\Gamma_2^{o}\right|=3$, and $a\in  \Gamma_2^{o}$ then $a \in \ob 8$. }
\end{minipage}
\caption{X-1-N addresses create orbits 3 to 8 where the first transaction has a single output address and the second transaction receives coins from the address. Red addresses have active orbits, whereas gray addresses have passive orbits.\label{fig:orbitsx1n}} 
\end{figure*}

\begin{figure*}[]
\begin{minipage}[b]{0.30\linewidth}
  \centering 
  \includegraphics[width=\ra\linewidth]{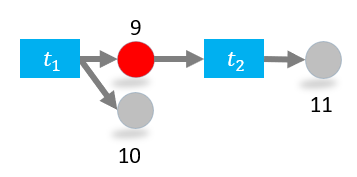} 
  \caption*{(g) Orbit 9:
if $\left|\Gamma_1^{o}\right|=2$, $\left|\Gamma_2^{i}\right|=1$, $\left|\Gamma_2^{o}\right|=1$, $a \in \Gamma_1^{o} $, $a \in \Gamma_2^{i} $ then $a \in\ob9$. Orbit 10:
if $\left|\Gamma_1^{o}\right|=2$, $\left|\Gamma_2^{i}\right|=1$, $\left|\Gamma_2^{o}\right|=1$, $a \in \Gamma_1^{o} $, $a \notin \Gamma_2^{i} $ then $a \in\ob10$.  
 Orbit 11: if $\left|\Gamma_1^{o}\right|=2$, $\left|\Gamma_2^{i}\right|=1$, $\left|\Gamma_2^{o}\right|=1$, $a \in \Gamma_2^{o} $ then $a \in\ob11$. }
\end{minipage} 
\hfill%
\begin{minipage}[b]{0.30\linewidth}
  \centering 
  \includegraphics[width=\ra\linewidth]{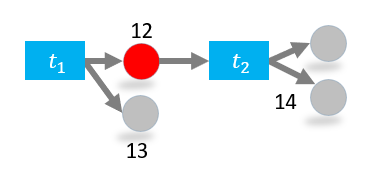} 
  \caption*{(h) Orbit 12:
if $\left|\Gamma_1^{o}\right|=2$, $\left|\Gamma_2^{i}\right|=1$, $\left|\Gamma_2^{o}\right|=2$, $a \in \Gamma_1^{o} $, $a \in \Gamma_2^{i} $ then $a \in\ob12$. Orbit   13:
if $\left|\Gamma_1^{o}\right|=2$, $\left|\Gamma_2^{i}\right|=1$,$\left|\Gamma_2^{o}\right|=2$, $a \in \Gamma_1^{o} $, $a \notin \Gamma_2^{i} $ then $a \in\ob13$. Orbit 14:
if $\left|\Gamma_1^{o}\right|=2$, $\left|\Gamma_2^{i}\right|=1$,$\left|\Gamma_2^{o}\right|=2$, $a \in \Gamma_2^{o} $ then $a \in\ob14$.}
\end{minipage}
\hfill%
\begin{minipage}[b]{0.30\linewidth}
  \centering 
  \includegraphics[width=\ra\linewidth]{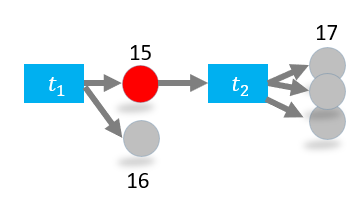} 
  \caption*{(i) Orbit 15:
if $\left|\Gamma_1^{o}\right|=2$, $\left|\Gamma_2^{i}\right|=1$,$\left|\Gamma_2^{o}\right|=3$, $a \in \Gamma_1^{o}$, $a \in \Gamma_2^{i} $ then  $a \in\ob15$.
 Orbit 16:
 if $\left|\Gamma_1^{o}\right|=2$, $\left|\Gamma_2^{i}\right|=1$,$\left|\Gamma_2^{o}\right|=3$, $a \in \Gamma_1^{o}$, $a \notin \Gamma_2^{i} $ then  $a \in\ob16$.
 Orbit   17:
if $\left|\Gamma_1^{o}\right|=2$, $\left|\Gamma_2^{i}\right|=1$,$\left|\Gamma_2^{o}\right|=3$, $a \in \Gamma_2^{o} $ then $a \in\ob17$. }
  \label{fig:dominatedOGB}
\end{minipage}
\\
\begin{minipage}[b]{0.30\linewidth}
  \centering 
  \includegraphics[width=\ra\linewidth]{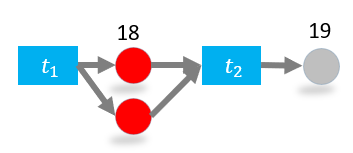} 
  \caption*{(j) Orbit   18:
if $\left|\Gamma_1^{o}\right|=2$, $\left|\Gamma_2^{i}\right|=2$,$\left|\Gamma_2^{o}\right|=1$, $a \in \Gamma_1^{o}$, $a \in \Gamma_2^{i} $ then  $a \in\ob18$.
 Orbit   19:
if $\left|\Gamma_1^{o}\right|=2$, $\left|\Gamma_2^{i}\right|=2$,$\left|\Gamma_2^{o}\right|=1$, $a \in \Gamma_2^{o}$, then  $a \in\ob19$.  }
\end{minipage}
\hfill%
\begin{minipage}[b]{0.30\linewidth}
  \centering 
  \includegraphics[width=\ra\linewidth]{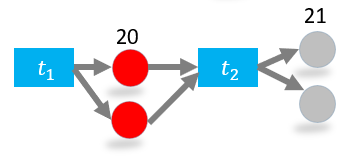} 
  \caption*{(k) Orbit   20:
if $\left|\Gamma_1^{o}\right|=2$, $\left|\Gamma_2^{i}\right|=2$,$\left|\Gamma_2^{o}\right|=2$, $a \in \Gamma_1^{o}$, $a \in \Gamma_2^{i} $ then  $a \in\ob20$. 
 Orbit   21:
if $\left|\Gamma_1^{o}\right|=2$, $\left|\Gamma_2^{i}\right|=2$,$\left|\Gamma_2^{o}\right|=2$, $a \in \Gamma_2^{o}$, then  $a \in\ob21$.   }
\end{minipage}
\hfill%
\begin{minipage}[b]{0.30\linewidth}
  \centering 
  \includegraphics[width=\ra\linewidth]{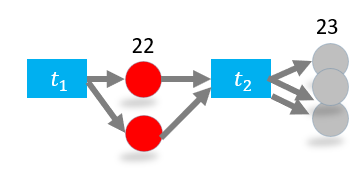} 
  \caption*{(l) Orbit   22:
if $\left|\Gamma_1^{o}\right|=2$, $\left|\Gamma_2^{i}\right|=2$,$\left|\Gamma_2^{o}\right|=3$, $a \in \Gamma_1^{o}$, $a \in \Gamma_2^{i} $ then  $a \in\ob22$. 
 Orbit   23:
if $\left|\Gamma_1^{o}\right|=2$, $\left|\Gamma_2^{i}\right|=2$,$\left|\Gamma_2^{o}\right|=3$, $a \in \Gamma_2^{o}$, then  $a \in\ob23$.  }
\end{minipage}
\caption{X-2-N create orbits 9 to 23 where the first transaction has two output addresses and the second transaction receives coins from one or two of the output addresses of the first transaction. Red addresses have active orbits, whereas gray addresses have passive orbits.\label{fig:orbitsx2n}} 
\end{figure*}

\begin{figure*}[]
\begin{minipage}[b]{0.30\linewidth}
  \centering 
  \includegraphics[width=\ra\linewidth]{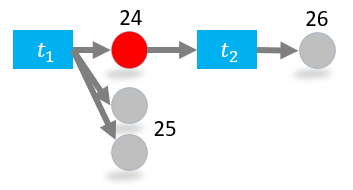} 
  \caption*{(m) Orbit   24:
if $\left|\Gamma_1^{o}\right|=3$, $\left|\Gamma_2^{i}\right|=1$,$\left|\Gamma_2^{o}\right|=1$, $a \in \Gamma_1^{o}$, $a \in \Gamma_2^{i} $ then  $a \in\ob24$.
 Orbit   25:
if $\left|\Gamma_1^{o}\right|=3$, $\left|\Gamma_2^{i}\right|=1$,$\left|\Gamma_2^{o}\right|=1$, $a \in \Gamma_1^{o}$, $a \notin \Gamma_2^{i} $ then  $a \in\ob25$.
 Orbit   26:
if $\left|\Gamma_1^{o}\right|=3$, $\left|\Gamma_2^{i}\right|=1$,$\left|\Gamma_2^{o}\right|=1$, $a \in \Gamma_2^{o}$, then  $a \in\ob26$. }
\end{minipage}
\hfill%
\begin{minipage}[b]{0.30\linewidth}
  \centering 
  \includegraphics[width=\ra\linewidth]{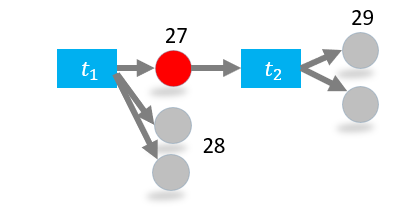} 
  \caption*{(n) Orbit   27:
if $\left|\Gamma_1^{o}\right|=3$, $\left|\Gamma_2^{i}\right|=1$,$\left|\Gamma_2^{o}\right|=2$, $a \in \Gamma_1^{o}$, $a \in \Gamma_2^{i} $ then  $a \in\ob27$.
 Orbit   28:
if $\left|\Gamma_1^{o}\right|=3$, $\left|\Gamma_2^{i}\right|=1$,$\left|\Gamma_2^{o}\right|=2$, $a \in \Gamma_1^{o}$, $a \notin \Gamma_2^{i} $ then  $a \in\ob28$.
 Orbit   29:
if $\left|\Gamma_1^{o}\right|=3$, $\left|\Gamma_2^{i}\right|=1$,$\left|\Gamma_2^{o}\right|=2$, $a \in \Gamma_2^{o} $ then  $a \in\ob29$.   }
\end{minipage}
\hfill%
\begin{minipage}[b]{0.30\linewidth}
  \centering 
  \includegraphics[width=\ra\linewidth]{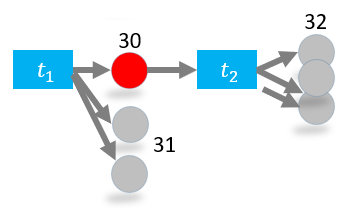} 
  \caption*{(o) Orbit   30:
if $\left|\Gamma_1^{o}\right|=3$, $\left|\Gamma_2^{i}\right|=1$,$\left|\Gamma_2^{o}\right|=3$, $a \in \Gamma_1^{o}$, $a \in \Gamma_2^{i} $ then  $a \in\ob30$.
 Orbit   31:
if $\left|\Gamma_1^{o}\right|=3$, $\left|\Gamma_2^{i}\right|=1$,$\left|\Gamma_2^{o}\right|=3$, $a \in \Gamma_1^{o}$, $a \notin \Gamma_2^{i} $ then  $a \in\ob31$.
 Orbit   32:
if $\left|\Gamma_1^{o}\right|=3$, $\left|\Gamma_2^{i}\right|=1$,$\left|\Gamma_2^{o}\right|=3$, $a \in \Gamma_2^{o} $ then  $a \in\ob32$.   } 
\end{minipage}
\\
\begin{minipage}[b]{0.30\linewidth}
  \centering 
  \includegraphics[width=\ra\linewidth]{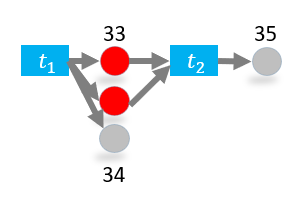} 
  \caption*{(p) Orbit   33:
if $\left|\Gamma_1^{o}\right|=3$, $\left|\Gamma_2^{i}\right|=2$,$\left|\Gamma_2^{o}\right|=1$, $a \in \Gamma_1^{o}$, $a \in \Gamma_2^{i} $ then  $a \in\ob33$.
 Orbit   34:
if $\left|\Gamma_1^{o}\right|=3$, $\left|\Gamma_2^{i}\right|=2$,$\left|\Gamma_2^{o}\right|=1$, $a \in \Gamma_1^{o}$, $a \notin \Gamma_2^{i} $ then  $a \in\ob34$.
 Orbit   35:
if $\left|\Gamma_1^{o}\right|=3$, $\left|\Gamma_2^{i}\right|=2$,$\left|\Gamma_2^{o}\right|=1$, $a \in \Gamma_2^{o}$ then  $a \in\ob35$. }
\end{minipage}
\hfill%
\begin{minipage}[b]{0.30\linewidth}
  \centering 
  \includegraphics[width=\ra\linewidth]{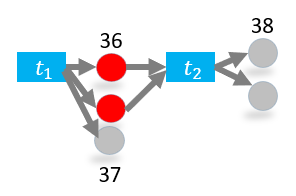} 
  \caption*{(r)  Orbit   36:
if $\left|\Gamma_1^{o}\right|=3$, $\left|\Gamma_2^{i}\right|=2$,$\left|\Gamma_2^{o}\right|=2$, $a \in \Gamma_1^{o}$, $a \in \Gamma_2^{i} $ then  $a \in\ob36$.
 Orbit   37:
if $\left|\Gamma_1^{o}\right|=3$, $\left|\Gamma_2^{i}\right|=2$,$\left|\Gamma_2^{o}\right|=2$, $a \in \Gamma_1^{o}$, $a \notin \Gamma_2^{i} $ then  $a \in\ob37$.
 Orbit   38:
if $\left|\Gamma_1^{o}\right|=3$, $\left|\Gamma_2^{i}\right|=2$,$\left|\Gamma_2^{o}\right|=2$, $a \in \Gamma_2^{o}$ then  $a \in\ob38$.}
\end{minipage}
\hfill%
\begin{minipage}[b]{0.30\linewidth}
  \centering 
  \includegraphics[width=\ra\linewidth]{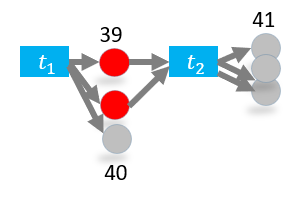} 
  \caption*{(s)  Orbit   39:
if $\left|\Gamma_1^{o}\right|=3$, $\left|\Gamma_2^{i}\right|=2$,$\left|\Gamma_2^{o}\right|=3$, $a \in \Gamma_1^{o}$, $a \in \Gamma_2^{i} $ then  $a \in\ob39$.
 Orbit   40:
if $\left|\Gamma_1^{o}\right|=3$, $\left|\Gamma_2^{i}\right|=2$,$\left|\Gamma_2^{o}\right|=3$, $a \in \Gamma_1^{o}$, $a \notin \Gamma_2^{i} $ then  $a \in\ob40$.
 Orbit   41:
if $\left|\Gamma_1^{o}\right|=3$, $\left|\Gamma_2^{i}\right|=2$,$\left|\Gamma_2^{o}\right|=3$, $a \in \Gamma_2^{o}$ then  $a \in\ob41$. }
\end{minipage}
\\
\begin{minipage}[b]{0.30\linewidth}
  \centering 
  \includegraphics[width=\ra\linewidth]{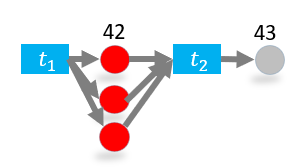} 
  \caption*{(t) Orbit   42:
if $\left|\Gamma_1^{o}\right|=3$, $\left|\Gamma_2^{i}\right|=3$,$\left|\Gamma_2^{o}\right|=1$, $a \in \Gamma_1^{o}$, $a \in \Gamma_2^{i} $ then  $a \in\ob42$.
 Orbit   43:
if $\left|\Gamma_1^{o}\right|=3$, $\left|\Gamma_2^{i}\right|=3$,$\left|\Gamma_2^{o}\right|=1$, $a \in \Gamma_2^{o}$ then  $a \in\ob43$. }
 
\end{minipage}
\hfill%
\begin{minipage}[b]{0.30\linewidth}
  \centering 
  \includegraphics[width=\ra\linewidth]{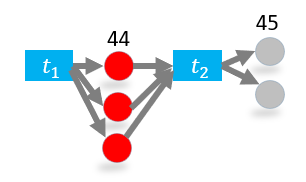} 
  \caption*{(u)  Orbit   44:
if $\left|\Gamma_1^{o}\right|=3$, $\left|\Gamma_2^{i}\right|=3$,$\left|\Gamma_2^{o}\right|=2$, $a \in \Gamma_1^{o}$, $a \in \Gamma_2^{i} $ then  $a \in\ob44$.
 Orbit   45:
if $\left|\Gamma_1^{o}\right|=3$, $\left|\Gamma_2^{i}\right|=3$,$\left|\Gamma_2^{o}\right|=2$, $a \in \Gamma_2^{o}$ then  $a \in\ob45$.   }
\end{minipage}
\hfill%
\begin{minipage}[b]{0.30\linewidth}
  \centering 
  \includegraphics[width=\ra\linewidth]{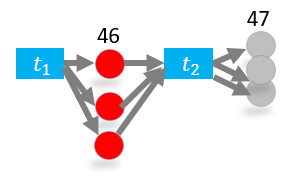} 
  \caption*{(v)  Orbit   46:
if $\left|\Gamma_1^{o}\right|=3$, $\left|\Gamma_2^{i}\right|=3$,$\left|\Gamma_2^{o}\right|=3$, $a \in \Gamma_1^{o}$, $a \in \Gamma_2^{i} $ then  $a \in\ob46$.
 Orbit   47:
if $\left|\Gamma_1^{o}\right|=3$, $\left|\Gamma_2^{i}\right|=3$,$\left|\Gamma_2^{o}\right|=3$, $a \in \Gamma_2^{o}$ then  $a \in\ob47$. }
\end{minipage}
 \caption{X-3-N create orbits 24 to 47 where the first transaction has three output addresses and the second transaction receives coins from one, two or three of the output addresses of the first transaction. Red addresses have active orbits, whereas gray addresses have passive orbits.\label{fig:orbitsx3n}} 
\end{figure*}

\section{Orbits and Coin-Mixing Transactions} 
\label{sec:mixing}
A coin-mixing transaction, also known as a coin-join transaction~\cite{maxwell2013coinjoin}, is a technique used to enhance the privacy and anonymity of cryptocurrency transactions. In a typical coin-mixing transaction, multiple participants combine their coins or funds into a single transaction, making it difficult to trace the original source of the funds.

A coin-mixing transaction works in the following way:

Participants: Several individuals who wish to improve the privacy of their transactions agree to participate in a coin-mixing process.

Inputs: Each participant selects a specific amount of their cryptocurrency holdings (coins) to contribute to the coin-mixing transaction. These individual inputs are usually of equal value to obfuscate the original ownership.

Mixing Pool: All the selected inputs from the participants are combined into a mixing pool or a single transaction. This process effectively mixes the coins together, making it challenging to identify which coins belong to each participant.

Outputs: The mixing pool generates new outputs or addresses that receive the mixed coins. The number of outputs is usually the same as the number of inputs to maintain balance.

Distribution: The newly mixed coins are then distributed back to the participants, but each participant receives their coins in a randomized manner. This step further obscures the link between the original inputs and the final outputs.

The mixing process is repeated multiple times to enhance privacy. It's worth noting that while coin-mixing transactions enhance privacy, they can also be associated with illicit activities if used for money laundering or other illegal purposes.  

In graph terms, the coin mixing transactions create 2-chainlets with many inputs and outputs. We identify several orbits that may encode coin mixing transactions. These are 
 orbits 30, 31, 32, 39, 40, 41, 46, and 47 in \Cref{fig:orbitsx2n,fig:orbitsx3n}. 

%% file: arxiv.bbl
\begin{thebibliography}{10}

\bibitem{akcora2017chainlet}
C.~G. Akcora, A.~K. Dey, Y.~R. Gel, and M.~Kantarcioglu.
\newblock Forecasting bitcoin price with graph chainlets.
\newblock In {\em The PAKDD, Melbourne, Australia}, pages 1--12, 2018.

\bibitem{akcora2021bitcoinheist}
Cuneyt~G Akcora, Yitao Li, Yulia~R Gel, and Murat Kantarcioglu.
\newblock Bitcoinheist: topological data analysis for ransomware prediction on
  the bitcoin blockchain.
\newblock In {\em Proceedings of the Twenty-Ninth International Conference on
  International Joint Conferences on Artificial Intelligence}, pages
  4439--4445, 2021.

\bibitem{akcora2021blockchain}
Cuneyt~Gurcan Akcora, Yulia~R Gel, and Murat Kantarcioglu.
\newblock Blockchain networks: Data structures of bitcoin, monero, zcash,
  ethereum, ripple, and iota.
\newblock {\em Wiley Interdisciplinary Reviews: Data Mining and Knowledge
  Discovery}, page e1436, 2021.

\bibitem{akcora2022blockchain}
Cuneyt~Gurcan Akcora, Yulia~R Gel, and Murat Kantarcioglu.
\newblock Blockchain networks: Data structures of bitcoin, monero, zcash,
  ethereum, ripple, and iota.
\newblock {\em Wiley WIRES}, 12(1):e1436, 2022.

\bibitem{Alarab}
Ismail Alarab, Simant Prakoonwit, and Mohamed~Ikbal Nacer.
\newblock Competence of graph convolutional networks for anti-money laundering
  in bitcoin blockchain.
\newblock In {\em Proceedings of the 2020 5th International Conference on
  Machine Learning Technologies}, pages 23--27, 2020.

\bibitem{androulaki2013evaluating}
Elli Androulaki, Ghassan~O Karame, Marc Roeschlin, Tobias Scherer, and Srdjan
  Capkun.
\newblock Evaluating user privacy in bitcoin.
\newblock In {\em IFCA}, pages 34--51. Springer, 2013.

\bibitem{gwerndarknet}
Gwern Branwen.
\newblock Darknet market archives.
\newblock Online, 2023.

\bibitem{breiman2001random}
Leo Breiman.
\newblock Random forests.
\newblock {\em Machine learning}, 45:5--32, 2001.

\bibitem{conti2018economic}
Mauro Conti, Ankit Gangwal, and Sushmita Ruj.
\newblock On the economic significance of ransomware campaigns: A bitcoin
  transactions perspective.
\newblock {\em Computers \& Security}, 2018.

\bibitem{erricafair20}
Federico Errica, Marco Podda, Davide Bacciu, and Alessio Micheli.
\newblock A fair comparison of graph neural networks for graph classification.
\newblock In {\em 8th International Conference on Learning Representations,
  {ICLR} 2020, Addis Ababa, Ethiopia, April 26-30, 2020}. OpenReview.net, 2020.

\bibitem{hamilton2017inductive}
Will Hamilton, Zhitao Ying, and Jure Leskovec.
\newblock Inductive representation learning on large graphs.
\newblock {\em Advances in neural information processing systems}, 30, 2017.

\bibitem{huang2018tracking}
D.~Y. Huang and D.~McCoy.
\newblock Tracking ransomware end-to-end.
\newblock In {\em Tracking Ransomware End-to-end}, pages 1--12. IEEE, 2018.

\bibitem{kethineni2018use}
Sesha Kethineni, Ying Cao, and Cassandra Dodge.
\newblock Use of bitcoin in darknet markets: Examining facilitative factors on
  bitcoin-related crimes.
\newblock {\em American Journal of Criminal Justice}, 43:141--157, 2018.

\bibitem{kipf2016semi}
Thomas~N Kipf and Max Welling.
\newblock Semi-supervised classification with graph convolutional networks.
\newblock {\em arXiv preprint arXiv:1609.02907}, 2016.

\bibitem{lang2012algebra}
Serge Lang.
\newblock {\em Algebra}, volume 211.
\newblock Springer Science \& Business Media, 2012.

\bibitem{lee2019toward}
Chaehyeon Lee, Sajan Maharjan, Kyungchan Ko, and James Won-Ki Hong.
\newblock Toward detecting illegal transactions on bitcoin using
  machine-learning methods.
\newblock In {\em International Conference on Blockchain and Trustworthy
  Systems}, pages 520--533. Springer, 2019.

\bibitem{liao2016behind}
Kevin Liao, Ziming Zhao, Adam Doup{\'e}, and Gail-Joon Ahn.
\newblock Behind closed doors: measurement and analysis of cryptolocker ransoms
  in bitcoin.
\newblock In {\em 2016 APWG symposium on electronic crime research (eCrime)},
  pages 1--13. IEEE, 2016.

\bibitem{lorenz2020machine}
Joana Lorenz, Maria~In{\^e}s Silva, David Apar{\'\i}cio, Jo{\~a}o~Tiago
  Ascens{\~a}o, and Pedro Bizarro.
\newblock Machine learning methods to detect money laundering in the bitcoin
  blockchain in the presence of label scarcity.
\newblock In {\em Proceedings of the First ACM International Conference on AI
  in Finance}, pages 1--8, 2020.

\bibitem{lorrain1971structural}
Francois Lorrain and Harrison~C White.
\newblock Structural equivalence of individuals in social networks.
\newblock {\em The Journal of mathematical sociology}, 1(1):49--80, 1971.

\bibitem{martin2018depth}
Alejandro Mart{\'\i}n, Julio Hernandez-Castro, and David Camacho.
\newblock An in-depth study of the jisut family of android ransomware.
\newblock {\em IEEE Access}, 6:57205--57218, 2018.

\bibitem{maxwell2013coinjoin}
Greg Maxwell.
\newblock Coinjoin: Bitcoin privacy for the real world.
\newblock In {\em Post on Bitcoin Forum}, 2013.

\bibitem{meiklejohn2013fistful}
S.~Meiklejohn, M.~Pomarole, G.~Jordan, D.~Levchenko, K.and~McCoy, G.~M Voelker,
  and S.~Savage.
\newblock A fistful of bitcoins: characterizing payments among men with no
  names.
\newblock In {\em IMC}, pages 127--140. ACM, 2013.

\bibitem{monamo2016unsupervised}
Patrick Monamo, Vukosi Marivate, and Bheki Twala.
\newblock Unsupervised learning for robust bitcoin fraud detection.
\newblock In {\em 2016 Information Security for South Africa (ISSA)}, pages
  129--134. IEEE, 2016.

\bibitem{nakamoto2008bitcoin}
Satoshi Nakamoto.
\newblock Bitcoin: A peer-to-peer electronic cash system, 2008.

\bibitem{nan2018bitcoin}
Lihao Nan and Dacheng Tao.
\newblock Bitcoin mixing detection using deep autoencoder.
\newblock In {\em 2018 IEEE Third international conference on data science in
  cyberspace (DSC)}, pages 280--287. IEEE, 2018.

\bibitem{narayanan2017obfuscation}
A.~Narayanan and M.~M{\"o}ser.
\newblock Obfuscation in bitcoin: Techniques and politics.
\newblock {\em arXiv preprint arXiv:1706.05432}, 2017.

\bibitem{paquet2018ransomware}
Masarah Paquet-Clouston, Bernhard Haslhofer, and Benoit Dupont.
\newblock Ransomware payments in the bitcoin ecosystem.
\newblock {\em arXiv preprint arXiv:1804.04080}, 2018.

\bibitem{pareja2020evolvegcn}
Aldo Pareja, Giacomo Domeniconi, Jie Chen, Tengfei Ma, Toyotaro Suzumura,
  Hiroki Kanezashi, Tim Kaler, Tao Schardl, and Charles Leiserson.
\newblock Evolvegcn: Evolving graph convolutional networks for dynamic graphs.
\newblock In {\em Proceedings of the AAAI Conference on Artificial
  Intelligence}, pages 5363--5370, 2020.

\bibitem{pham2016anomaly}
Thai Pham and Steven Lee.
\newblock Anomaly detection in bitcoin network using unsupervised learning
  methods.
\newblock In {\em arXiv preprint arXiv:1611.03941}, 2016.

\bibitem{ruffing2014coinshuffle}
Tim Ruffing, Pedro Moreno-Sanchez, and Aniket Kate.
\newblock Coinshuffle: Practical decentralized coin mixing for bitcoin.
\newblock In {\em European Symposium on Research in Computer Security}, pages
  345--364. Springer, 2014.

\bibitem{vassallo2021application}
Dylan Vassallo, Vincent Vella, and Joshua Ellul.
\newblock Application of gradient boosting algorithms for anti-money laundering
  in cryptocurrencies.
\newblock In {\em SN Computer Science}, volume~2, pages 1--15. Springer, 2021.

\bibitem{AML}
Mark Weber, Giacomo Domeniconi, Jie Chen, Daniel Karl~I. Weidele, Claudio
  Bellei, Tom Robinson, and Charles~E. Leiserson.
\newblock Anti-money laundering in bitcoin: Experimenting with graph
  convolutional networks for financial forensics.
\newblock In {\em ACM SIGKDD International Workshop on Knowledge discovery and
  data mining}, 2019.

\bibitem{xiao2022graph}
Shunxin Xiao, Shiping Wang, Yuanfei Dai, and Wenzhong Guo.
\newblock Graph neural networks in node classification: survey and evaluation.
\newblock {\em Machine Vision and Applications}, 33:1--19, 2022.

\bibitem{xu2018powerful}
Keyulu Xu, Weihua Hu, Jure Leskovec, and Stefanie Jegelka.
\newblock How powerful are graph neural networks?
\newblock {\em arXiv preprint arXiv:1810.00826}, 2018.

\bibitem{zhang2018end}
Muhan Zhang, Zhicheng Cui, Marion Neumann, and Yixin Chen.
\newblock An end-to-end deep learning architecture for graph classification.
\newblock In {\em Proceedings of the AAAI conference on artificial
  intelligence}, volume~32, 2018.

\bibitem{zola2019cascading}
Francesco Zola, Maria Eguimendia, Jan~Lukas Bruse, and Raul~Orduna Urrutia.
\newblock Cascading machine learning to attack bitcoin anonymity.
\newblock In {\em 2019 IEEE International Conference on Blockchain
  (Blockchain)}, pages 10--17. IEEE, 2019.

\end{thebibliography}
